\documentclass[11pt]{IEEEtran2e} %
\pdfoutput=1
\usepackage{times}
\usepackage{amsmath}
\usepackage{subfigure}
\usepackage{colordvi}
\usepackage{subfigure}
\usepackage{epsfig}
\usepackage{amsfonts}
\usepackage{graphicx}
\usepackage{latexsym}
\usepackage{delarray}

\newcommand{\urlBiBTeX}[1]{\url{#1}}
\newcommand{\urlbibteX}[1]{\url{#1}}
\def\BibTeX{{\rm B\kern-.05em{\sc i\kern-.025em b}\kern-.08em
    T\kern-.1667em\lower.7ex\hbox{E}\kern-.125emX}}

\def\conv{*}
\def\deconv{\oslash}

\def\Smin{\underline{S}}
\def\Smax{\overline{S}}
\def\Four{{\mathcal{F}}}
\def\Legendre{{\mathcal{L}}}

\newtheorem{lemma}{Lemma}

\begin{document}
\title{A System Theoretic Approach to Bandwidth Estimation}

\author{
J\"{o}rg Liebeherr, ~Markus Fidler,
~Shahrokh Valaee
\vspace{-4mm}
\thanks{
J. Liebeherr and S. Valaee are with the Department of Electrical and Computer Engineering, University of Toronto. M. Fidler is with the Technical University of Darmstadt.
}
\thanks{
The  research in this paper is supported in part by the National
Science Foundation under grants CNS-0435061,  two NSERC Discovery
grants, and an Emmy Noether grant from the German Research
Foundation.   }
}
\maketitle

\begin{abstract}
It is shown that bandwidth estimation in packet networks can be viewed in  terms of min-plus linear system theory. The available bandwidth of a link or complete path is expressed in terms of a {\em service curve}, which is a function that appears in the network calculus to express the service available to a traffic flow.
The service curve is estimated based on measurements of a sequence of probing packets or passive measurements of  a sample path of arrivals.  It is shown that existing bandwidth estimation methods can be derived in the min-plus algebra of the network calculus, thus providing further mathematical justification for these methods.
Principal difficulties of   estimating available bandwidth from measurement of network probes are related to potential non-linearities
of the underlying network. When networks are viewed as systems that operate either in a linear or in a non-linear regime, it is  argued that probing schemes extract the most information at a point when the network crosses from a linear to a non-linear regime. Experiments on the Emulab testbed  at the University of Utah evaluate the robustness of the system theoretic interpretation of networks in practice. Multi-node experiments evaluate how well the  convolution operation of the min-plus algebra
provides estimates for the available bandwidth of a path from estimates of individual links.

\end{abstract}

\section{Introduction}

The benefits of knowing how much network bandwidth is available to
an application has motivated the development of techniques that
infer bandwidth availability from traffic
measurements \cite{akella:bfind,cprobe,dovrolis:packetdispersion,hu:IGI,jain:slops,jain:pathload,jain:pathvar,kapoor:capprobe,topp,paxson:measurements,ribeiro:pathchirp,strauss:spruce}.
With a large number of methods available  and much empirical experience gained, recently an increasing effort has been
put towards improving the theoretic understanding of measurement based estimation of available bandwidth, e.g., \cite{baccelli06,loguinov04,loguinov08,BinTariq}.

This paper presents a new foundational approach to reason about
available bandwidth estimation as the analysis of a  min-plus
linear system. Min-plus linear system theory has provided the
mathematical underpinning for the deterministic network calculus
\cite{chang:performanceguarantees,leboudec:networkcalculus}. We will use min-plus system
theory to explain how
bandwidth estimation methods infer information about a network
and find bandwidth estimation methods that can extract the most
information from a network. Some key difficulties encountered when
measuring available bandwidth become evident in a system theoretic
view.

We view bandwidth estimation as the problem of determining
unknown functions that describe the available bandwidth based on
measurements of a sequence of probing packets or passive
measurements of  a sample path of arrivals. These functions
correspond to the {\em service curves} that appear in
the network calculus \cite{cruz:qualityofserviceguarantees}, where
they are used to express the available service at a network link or
an end-to-end path. Working within the context of the network
calculus, we can apply a result that allows us to compute the
service curve of a network path from service curves of
the links of the path. This is done by applying the convolution
operator of the min-plus algebra
\cite{chang:performanceguarantees,leboudec:networkcalculus}. We
explore how well  the convolution of the available bandwidth of
multiple links, expressed as service curves, can describe
the available bandwidth of an end-to-end path.

Our  formulation of available bandwidth estimation in min-plus linear system theory reveals that the underlying problem is intrinsically hard, requiring the solution to  a maximin optimization problem.
The optimization problem becomes  more tractable when the network satisfies the property of `min-plus
linearity'. We show that some existing estimation techniques can be
accurately characterized if we interpret them as analyzing a network with linear input-output relationships.
The discovery of an implicit assumption of min-plus linearity in existing measurement methods is seemingly at odds with empirical evidence that these methods have been successfully applied in networks that do not satisfy linearity. For example, even a  single FIFO link violates the requirements of min-plus linearity.
We resolve this  apparent contradiction by showing that some networks
can be decomposed into disjoint min-plus linear and
non-linear regions. These networks behave as  a min-plus linear system at low load, and become non-linear  if the load exceeds a certain threshold.
The crossing of the linear and non-linear regions marks the point  where the available bandwidth  can be observed.

The arguments in this paper draw  from  known relationships
between linear system theory and the network calculus. The success
in describing  relatively complex probing schemes using min-plus
algebra and the ability to concatenate the available bandwidths of  multiple links using the min-plus convolution
hints at a possibly stronger link
between bandwidth estimation and network calculus.

The assumptions in this paper on network and traffic  characteristics
are analogous to those in most papers on bandwidth estimation techniques (see Section~\ref{sec-probe}).
The available bandwidth is represented by a  random process, where the source of randomness
is the variability of network traffic. A major assumption is that the time scale of network measurements is
small compared to the time scale at which characteristics of  network traffic or network links change.
This assumption is not justified when properties of a network link vary on short time scales, e.g., on wireless
transmission channels with random noise. Consequently, such networks are not  adequately described in our min-plus system theoretic formulation.

The objective  of this paper is to offer an alternative  interpretation for bandwidth estimation, that potentially enables the development of improved bandwidth estimation schemes.
We previously mentioned that the convolution operator in the min-plus algebra can be exploited to compute bandwidth estimates for end-to-end paths.
Additionally, by generalizing the available bandwidth in terms of service curves  we can express multiple data rates at different time scales. This makes it possible to distinguish a short-term reduction of the data rate due to temporary link congestion from the long-term  utilization of a link or a path.
While we discuss and evaluate implementations of bandwidth probing schemes in measurement experiments on a testbed network,
we emphasize that our objective is a validation of the system theoretic interpretation of these methods, and not an empirical comparison of existing probing schemes.

The remainder of this paper is structured as follows. In
Section~\ref{sec-probe}, we discuss bandwidth estimation methods and
other related work. In Section~\ref{sec-LTI}, we review  the min-plus
linear system interpretation of the deterministic network
calculus.
 In Section~\ref{sec-probetheory}, we formulate bandwidth estimation
 as the solution to an inversion problem in min-plus algebra.
 In Section~\ref{sec-inverse}, we derive solutions to
compute the inversion, and relate them to probing schemes
from the literature.  In
Section~\ref{sec-fifo}, we justify how these probing schemes can be
applied in networks that are not min-plus linear. In Section~\ref{sec-emulab}, we present measurement experiments of probing schemes suggested by the   min-plus system theoretic concepts from this paper. We present brief
conclusions in Section~\ref{sec-concl}.

\section{Available Bandwidth Estimation Techniques}
\label{sec-probe}

The goal of bandwidth estimation is to infer from measurements a
reliable estimate of the unused capacity at a multi-access link,
a single switch, or a network path.
The available bandwidth of a network link $i$ in a time interval
$[t, t+\tau)$ can be specified as \cite{strauss:spruce}
\[
\alpha_i (t, t+\tau) = \frac{1}{\tau} \int_{t}^{t+\tau} C_i (x) - \lambda_i (x) dx \ ,
\]
where $C_i (t)$ and $\lambda_i (t)$ are the capacity
and total traffic, respectively, on link~$i$ at time $t$.
We note that individual definitions of available bandwidth used in the literature may deviate from the above definition. It is generally assumed that  link capacities have a constant rate, i.e., $C_i (x) = C_i$.
Then, the  available bandwidth can be interpreted as a random process, where the randomness stems from the variability of network traffic.

If available bandwidth estimates for single links are available, the available bandwidth of an end-to-end network path with $H$ links
is computed as  \cite{jain:bandwidthestimationpitfalls}
 \begin{equation}
\alpha  (t, t+\tau) = \min_{i = 1, \ldots, H}  \alpha_i (t, t+\tau) \ .
\label{eq:alpha-new}
\end{equation}
The link at which the minimum is attained is often referred to as the {\em tight link}.
Available bandwidth methods measure the transmission of a sequence  of control (probe) packets
and use the measurements to estimate or  bound the available bandwidth. Closely related are probing
schemes that seek to determine the minimum capacity along a path,
referred to as {\em bottleneck capacity} or {\em capacity of the narrow link}.
If the time scale of  measurements is
small compared to the time scale at which characteristics of  network traffic changes, network traffic can be described by a deterministic function or even constant rate function.
In this case, a single sample of the available bandwidth can be interpreted as being  conditioned on the state of the network; Evaluating a large number of samples corresponds to  computing a conditional average.
Under a broad set of assumptions, such as stationarity of the distribution of traffic, the conditional averages are computed correctly.
When  network characteristics change on a short time scale, e.g., a  wireless  channels with random noise, a description of traffic
and link by deterministic functions is not suitable.

Almost all proposed probing schemes perform measurements of packet pairs or packet trains.
Packet pairs consist of two
packets with a defined spacing, and packet trains consist of more
than two packets. Since it was first suggested in
~\cite{jacobson88,keshav91}, packet pair probing has evolved
significantly, and has been used for estimating the bottleneck capacity
(e.g., {\em Bprobe}~\cite{cprobe}, {\em
CapProbe}~\cite{kapoor:capprobe}), the available
bandwidth (e.g., {\em ABwE} \cite{ABwE}, {\em
Spruce}~\cite{strauss:spruce}), and the distribution of
cross traffic~\cite{baccelli07}. The rationale behind these
methods builds on the relation of packet dispersion and available
bandwidth resources, i.e., packet pairs with a defined gap may be
spaced out on slow or loaded links and thus carry information
about the network path. Some techniques, e.g.,
\cite{baccelli07,strauss:spruce} build on a model of a single link whose capacity is assumed to be known.

The majority of proposed methods employ packet trains for
bottleneck capacity estimation (e.g., {\em
PBM}~\cite{paxson:measurements}, {\em Cprobe}~\cite{cprobe}, {\em
pathrate}~\cite{dovrolis:packetdispersion}), and for available
bandwidth estimation (e.g., {\em
pathload}~\cite{jain:pathload,jain:slops}, {\em
pathvar}~\cite{jain:pathvar}, {\em TOPP}~\cite{topp}, {\em
PTR/IGI}~\cite{hu:IGI}, {\em pathchirp}~\cite{ribeiro:pathchirp},
and {\em BFind}~\cite{akella:bfind}). The general approach is to
adaptively vary the rate of probing traffic to induce congestion
in the network. A comprehensive discussion of all techniques is
beyond the scope of this paper. For details and empirical
evaluations of packet train and packet pair methods we
refer to a series of available articles
\cite{jain:bandwidthestimationpitfalls,shriram07,shriam05,sommers,strauss:spruce}.
Some studies have found that packet trains provide more reliable bandwidth estimates than packet pairs \cite{jain:bandwidthestimationpitfalls,loguinov07}.
The wide spectrum of bandwidth estimation methods indicates the
complexity of measuring available bandwidth in a network. In particular, the comparative evaluations of bandwidth estimation
methods sometimes widely disagree in their conclusions on the capabilities and limitations of individual methods.

For the purposes of this paper, the two packet train methods {\em pathload} and {\em pathchirp} are particularly relevant.
{\em Pathload} uses a sequence of constant rate packet trains, where the  transmission rate of consecutive trains is iteratively
varied until it converges to
the available bandwidth.
In {\em pathchirp},  the rate is varied within a single packet train  using  geometrically decreasing inter-packet gaps. Both methods interpret increasing delays as an indication of overload, i.e. to detect if the probing rate exceeds the available bandwidth.

Most estimation techniques  are designed with an assumption that the network as a whole exhibits the behavior of a single link
with constant rate fluid cross traffic. Often it is assumed that the network behaves as
a single FIFO system~\cite{hu:IGI,loguinov04,liu:packetpairdispersion,loguinov07,loguinov08,baccelli07,topp,melander:fcfsprobing,ribeiro:pathchirp,strauss:spruce}. This  is justified by the particular packet dispersion of FIFO systems which
is matched by empirical data~\cite{melander:fcfsprobing}. It has been found that the best
estimates are obtained if the probing traffic increases the load
close to, but not beyond, an overloaded state.

Some probing methods suggest that probing traffic should follow a
Poisson process
\cite{loguinov04,padhye04,loguinov07,paxson:measurements,strauss:spruce,yzhang},
since it can benefit from the PASTA (Poisson Arrivals See Time
Averages) property. Briefly, the PASTA property states that, under
a broad set of assumptions, a Poisson arrival process observes the
average state of the system. An empirical study \cite{BinTariq}
found that Poisson probing does not necessarily lead to improved
estimates of the available bandwidth.
Also, \cite{baccelli06} points out that in case of
non-intrusive probing, Poisson probing is not justified by
default and may even be inferior to other schemes, since it does
not minimize estimation variance nor does it provably reduce
inversion bias, e.g. when deriving quantities of interest such as
available bandwidths from observations.

A set of  analytical studies   \cite{liu:packetpairdispersion,loguinov07,loguinov08}
characterizes the dispersion of  probing traffic  over single hop and multi hop paths in terms of probing-response curves, and extracts the available bandwidth from these curves.
Under the assumption of fluid constant rate cross-traffic probing-response curves feature a sharp bend at the available bandwidth that is used as criterion by some methods, e.g. TOPP \cite{topp}. The mode of operation of many other methods, e.g. the detection of overload by {\em pathload}, can be related to these curves \cite{loguinov07}. Under general bursty cross-traffic the unique turning point of probing-response curves diminishes, whereas it can be recovered under idealized conditions, e.g. using packet trains of infinite length, as shown in \cite{liu:packetpairdispersion,loguinov07,loguinov08}.

An alternative approach to sending  probe packets is to obtain estimates of the available bandwidth through  passive measurements of user traffic. This is the preferred approach in measurement based
admission control (MBAC), which seeks to determine if a
network has sufficient resources to support minimal service
requirements for a traffic flow or aggregate \cite{cetinkaya:egressadmissioncontrol,Jiang05}. In
comparison to passive measurements, probing schemes have an
additional degree of freedom since they can control the
traffic profile of probing packets.

We  note that links between network calculus and bandwidth estimation have been made before  mostly in the context of MBAC
\cite{cetinkaya:egressadmissioncontrol,Jiang05,valaee:adhocadmissioncontrol,Wu03}. Since MBAC studies are set in a context of providing service guarantees, they generally seek to obtain a worst-case description of the available service or traffic, in terms of  time-invariant envelope functions. Worst-case characterizations, even if relaxed to stochastic bounds, tend to be highly conservative. In this paper, we do not use envelopes to describe traffic or service.
For traffic that is transmitted at a lower priority as in
\cite{Jiang05,valaee:adhocadmissioncontrol}, the network calculus  permits a concise description of the  available bandwidth as the leftover capacity which is unused by higher priority traffic.
Aspects of a  min-plus system theoretic interpretation of available bandwidth can be found in  \cite{agharebparast:slopedomain}, which
exploits a known relationship between the Legendre transform of the backlog and the available bandwidth.

\section{Min-Plus Linear System Theory for Networks}
\label{sec-LTI}

This section reviews the linear system representation of networks
and introduces needed concepts and notation.
We consider a continuous-time setting.

\begin{figure}
\centering \epsfig{file=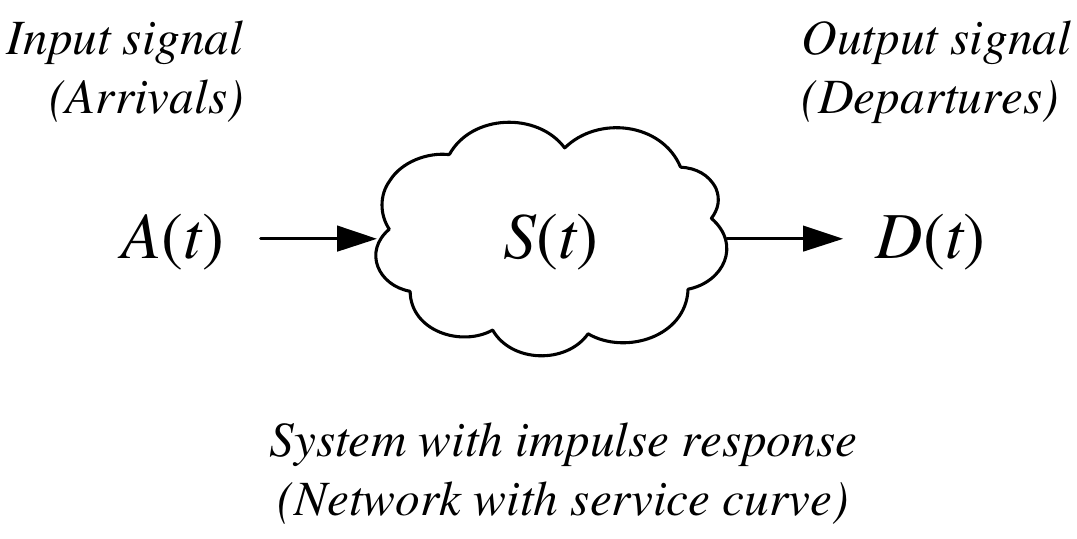, width=0.7 \linewidth}
\caption{Linear time-invariant system and min-plus linear
network.} \label{fig:LTIsystem}
\end{figure}

Classical linear system theory deals with linear time-invariant
(LTI) systems with input signal $A(t)$ and output signal $D(t)$ (see
Fig.~\ref{fig:LTIsystem}). Linear means that for any two pairs of
input and output signals $(A_1, D_1)$ and $(A_2, D_2)$, any linear
combination of input signals $b_1 A_1(t) + b_2 A_2(t)$
results in the linear combination of output signals $b_1
D_1(t) + b_2 D_2(t)$. Time-invariant means that for any pair of
inputs and outputs $(A, D)$, a time-shifted input $A(t-\tau)$
results in a shifted output $D(t-\tau)$.

Let $S(t)$ be the impulse response of the system, that is, the output
signal generated by the system if the input signal is a unity
(Dirac) impulse at time zero. The basic property of an LTI system is
that it is completely characterized by its impulse response, where
the output of the system is expressed as the convolution of the
input signal and the impulse response:
\begin{equation*}
D(t) = \int_{-\infty}^{\infty} A(\tau)  S(t-\tau) d\tau =: A \ast S
(t).
\end{equation*}
\subsection{Min-Plus Algebra in the Network Calculus}
A  significant discovery of networking research from the 1990's is
that networks can often be viewed as linear systems, when the usual
algebra is replaced by a so-called min-plus algebra
\cite{agrawal:flowcontrolprotocols,chang:performanceguarantees,leboudec:networkcalculus}. In a min-plus algebra
\cite{baccelli:synchronizationlinearity},  addition is replaced by a
minimum (we write infimum) and multiplication is replaced by an
addition. Similar to LTI systems, a min-plus linear system is a
system that is linear under the min-plus algebra. This means that a min-plus
linear combination of input functions $\inf \{ b_1 + A_1(t), b_2 +
A_2(t) \}$ results in the corresponding linear combination of output signals
$\inf \{ b_1 + D_1(t), b_2 + D_2(t) \}$.
In min-plus system theory, the burst function
\begin{equation}
\delta(t) =
\begin{cases}
\infty \ ,  & \text{if } t > 0 \ ,\\
0 \ ,  & \text{otherwise} \ ,
\end{cases}
\label{eq:burstfunction}
\end{equation}
takes the place of the Dirac impulse function.

Let $S(t)$ be the impulse response, that is, the output when the
input is the burst function $\delta(t)$. Any time-invariant min-plus
linear system is completely described by its impulse response, and
the output of any min-plus linear system can be expressed as a
linear combination of the input and shifted impulse responses
by
\begin{equation*}
D(t) = \inf_{\tau} \{A(\tau) + S(t-\tau)\}  =: A \ast S (t).
\label{eq:minplusconvolution}
\end{equation*}
In analogy to LTI systems, this operation is referred to as
convolution of the min-plus algebra
\cite{baccelli:synchronizationlinearity}.\footnote{We re-use the
symbol of the operator for notational simplicity. The context makes this slight
abuse of notation non-ambiguous.}
If there exists a
function $S(t)$ such that $D(t) = A \ast S(t)$ for all
pairs $(A, D)$, then it follows that the system is min-plus
linear.

The min-plus convolution shares many properties with the
usual convolution, e.g., it is commutative and associative.
The associativity of min-plus convolution is of particular
importance since it implies an easy way of
concatenating  systems in series. Given a tandem of two min-plus
linear systems $S_1(t)$ and $S_2(t)$, the output can be computed
iteratively as $D(t) = (A \conv S_1) \conv S_2(t)$ and, with
associativity, $D(t) = A \conv (S_1 \conv S_2)(t)$ holds. Generalizing, a tandem of $N$ systems that are characterized by
impulse responses $S_1, S_2, \ldots , S_N$
is equivalent to a single  system with impulse response
\begin{equation}
S(t) = S_1 \conv S_2 \conv \ldots \conv S_N (t) \ .
\label{eq-conv}
\end{equation}

The observation that some networks can be adequately modeled by a
min-plus linear system led to the min-plus formulation of the
network calculus \cite{agrawal:flowcontrolprotocols,chang:performanceguarantees,leboudec:networkcalculus}. Here, a
system is a network element or entire network, input and output
functions $A$ and $D$ are arrivals and departures, respectively, and
the impulse response $S$, called the {\em service curve}, represents
the service guarantee by a network element. Network elements  that
are known to be min-plus linear include  work-conserving constant
rate links ($S(t)=C \, t$, where $C$ is the link capacity), traffic shapers
($S(t)=\sigma + \rho \, t$, where $\sigma$ is a burst size and
$\rho$ is a rate), and rate-latency servers ($S(t)= r \, (t -d)_+$,
where $r$ is a rate, $d$ is a delay, and $(x)_+=\max (x,0)$), and
their concatenations. As in~\cite{agrawal:flowcontrolprotocols,chang:performanceguarantees,leboudec:networkcalculus} we make the convention that functions in the min-plus linear system theory are non-decreasing
non-negative functions that pass through the origin.

The relevance of the network calculus as a tool for the analysis of
networks results from an extension of its formal framework to
networks that do not satisfy the conditions of min-plus linearity.
Non-linear systems implement more complex mappings $\Pi$ of arrival
to departure functions $D(t) = \Pi(A)(t)$. In the network calculus,
these are replaced by linear mappings that provide bounds of the
form $D(t) \ge A \ast \Smin (t)$ or $D(t) \le A \ast \Smax (t)$
(\cite{leboudec:networkcalculus}, pp. {\em xviii}).
Here, $\Smin$
is referred to as a {\em lower service curve} and $\Smax$ is
referred to as an {\em upper service curve}, indicating that they are
bounds on the available service. In a min-plus linear system, the
service  curve $S$ is both an upper and a lower service curve
($S=\Smin=\Smax$), which is therefore frequently referred to as {\em  exact
service curve}.

\subsection{Legendre transform in Min-Plus Linear Systems}
\label{subsec-legendre}
In classical linear system theory, the Fourier transform of $f(t)$, denoted
by $\Four_f(\omega)$, establishes a dual domain, the frequency
domain, for analysis of LTI systems. In the frequency domain, the
Fourier transform turns the  convolution to a multiplication, that
is, $\Four_{f \ast g}(\omega) = \Four_f (\omega) \cdot
\Four_g(\omega)$.

In min-plus linear systems, the {\em Legendre transform}, also
referred to as convex Fenchel conjugate, plays a similar role. The
Legendre transform of a function $f(t)$ is defined as
\begin{equation*}
\Legendre_f (r) = \sup_{\tau} \{r\tau - f(\tau) \}.
\end{equation*}
Since $r$ can be interpreted as a rate, one may view the domain
established by the Legendre transform as a rate domain. The Legendre
transform takes the min-plus convolution to an
addition~\cite{baccelli:synchronizationlinearity,rockafellar:convexanalysis},
that is,
\footnote{Whenever
possible, from now on we use the shorthand notation $f$ to mean
`$f (t)$ for all $t \geq 0$', and $\Legendre_f$ to
mean `$\Legendre_f (r)$ for all $r \geq 0$'.}
\begin{equation}
\Legendre_{f \ast g} = \Legendre_f + \Legendre_g \ .
\label{eq:leg-add}
\end{equation}
Other properties of the Legendre transform that we exploit in this
paper are that, for convex functions $f$, we have
\begin{equation}
\Legendre(\Legendre_f) = f \ .
\label{eq:leg-convex}
\end{equation}
In other words, a convex function $f$ can be recovered from
$\Legendre_f$ by reapplying the Legendre transform
\cite{rockafellar:convexanalysis}. In general, we only have
\begin{equation}
\Legendre(\Legendre_f) \le f  \quad \mbox{ and } \quad
\Legendre(\Legendre_f) = \text{conv}_f \ , \label{eq:leg-nonconvex}
\end{equation}
where $\text{conv}_f$ denotes the convex hull of $f$, defined as the
largest convex function smaller than $f$.

Another property that will be used is that the
Legendre transform  reverses the order of an inequality, i.e.,
\begin{equation}
f \ge g  \Rightarrow \Legendre_f \le \Legendre_g  \ .
\label{eq:leg-rev}
\end{equation}
The  statement is an equivalency when $g$ is convex.
Applications of the Legendre transform in the network calculus
have been previously studied in \cite{agharebparast:slopedomain,fidler:legendre,hisakado:legendre,naudts:trafficparametermeasurement}.

\section{A Min-Plus Algebra Formulation of the Bandwidth Estimation Problem}
\label{sec-probetheory}

We view a network as a  min-plus linear or non-linear
system that converts input
signals (arrivals) into output signals (departures) according to a
fixed but unknown service curve $S$.
The service curve of the network expresses the available bandwidth,
which can be a constant-rate or a more complex function.
Measurements of a
network probe, defined as a sequence of at least two packets, can be
characterized by an arrival function $A^p(t)$ and a departure
function $D^p(t)$, where the functions represent the cumulative
number of bits seen in the interval $[0,t]$ and time~0 denotes the beginning of the probe.
We assume that the system satisfies time-invariance over the duration of a probe. This corresponds to an assumption stated in Section~\ref{sec-probe} that network characteristics do not change over the duration of a measurement.
The arrival and departure
functions of a probe are constructed from timestamps of the transmission and
reception of packets, and from knowledge of the packet size. In
Fig.~\ref{fig-probe} we illustrate a network probe consisting of
five packets of equal size with fixed spacing between
consecutive packets. The vertical distance between
arrivals and departures is defined as the virtual backlog $B^p(t)= A^p(t) - D^p(t)$.
The horizontal distance is defined as the virtual delay $W^p(t)$.

\begin{figure}
\centering \epsfig{file=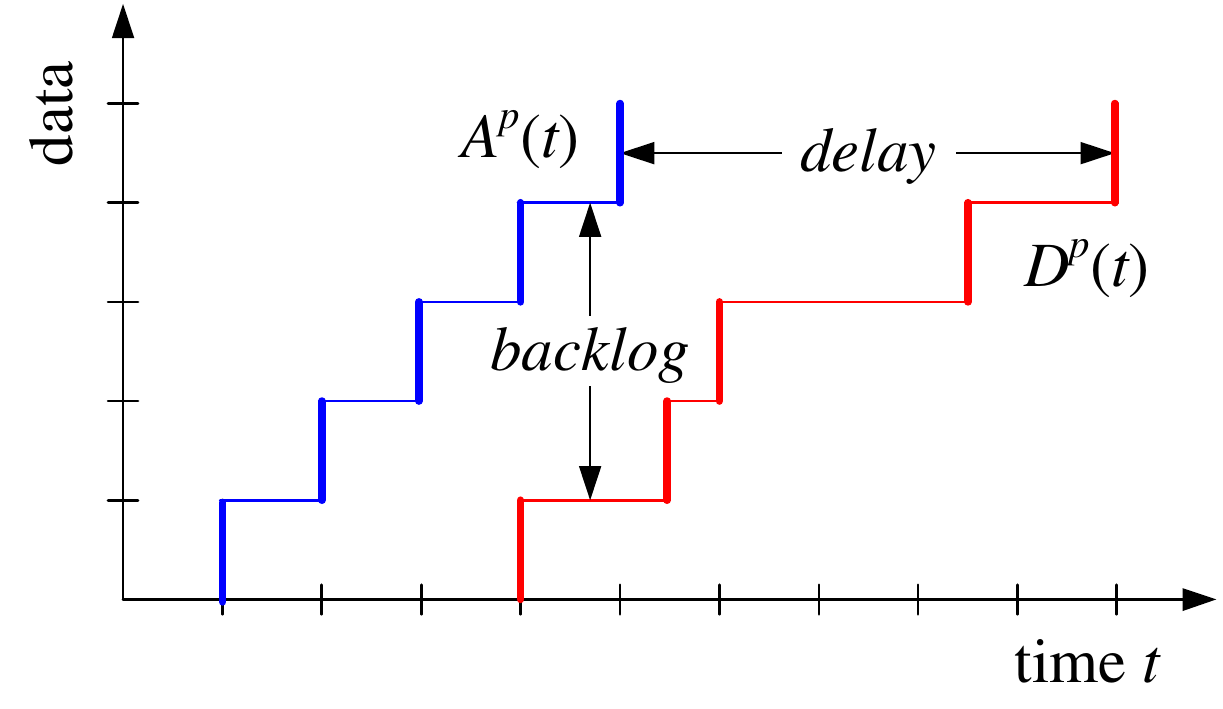, width=0.75 \linewidth}
\caption{Example arrival and departure function of a probe of five
packets.}\label{fig-probe}
\end{figure}

Representing the network by a min-plus linear system, we interpret a probing scheme as trying to determine from a specific sample of
functions $A^p$ and $D^p$ an estimate of an unknown lower service
$\Smin$, such that  $D\geq A \conv \Smin$ holds for {\sl all}
pairs $(A,D)$ of arrival and departure functions.
Ideally, the estimate should
be  a maximal $\Smin (t)$, i.e., there is no other lower service curve larger than $\Smin (t)$ that satisfies the definition.\footnote{We define a partial ordering of functions such that
$f \leq  g$ iff. $f (t) \leq g (t)$ for all $t$.} The goal of a probing scheme  is to select a probing
pattern, i.e., a function $A^p$, that reveals a
maximal service curve.
A maximal lower service curve $\Smin$ computed from $A^p$ and $D^p$ yields a sample of the available bandwidth.

Putting these considerations into a problem formulation, $\Smin$ is the solution to the following optimization problem:
\begin{tabbing}
\= xxxxxxxxxxxxx \=xxxxxxxxxxxx \= \kill
\>{\sc maximize} \> $\Smin$ \\
\>{\sc subject to} \>   $D(t) \geq  \inf_{\tau}
\{A(\tau) + \Smin (t-\tau)\}$, \\
\>\> \hspace{1cm} $\forall t \geq 0, \mbox{ for all pairs } (A, D)$.
\end{tabbing}
This problem has the structure of a maximin optimization, a class of problems which  is fundamentally hard.
The formulation does not consider that service curves only form a partial ordering. Therefore, there may not be an optimal solution, but only solutions  that cannot be further improved.

The bandwidth estimation problem is easier when the network can be
described by a min-plus linear system.
As we will see in Section~\ref{sec-fifo}, some
non-linear networks, such as FIFO systems,
are min-plus linear under low load conditions. Recalling that a system is min-plus linear
if it can be described by an exact service curve,  the bandwidth
estimation problem is reduced to solving the inversion of
$$
D(t) = A \conv S (t) \, \mbox{for all } t \geq 0.
$$

If we can take a measurement of $A^p$  and $D^p$ which solves the
equation for $S$, then, due to min-plus linearity, we have a
solution for {\sl all} possible arrival and departure functions. From
Section~\ref{sec-LTI}, we can infer that a solution is obtained
by using the burst function of Eq.~(\ref{eq:burstfunction}) as
probing pattern, i.e., $A^p (t) = \delta(t)$.
This follows since the service curve is the impulse response of
a min-plus system, that is, $D^p(t) = \delta \conv S (t) = S(t)$.
However, sending a probe as a burst
function is not practical, since it assumes the instantaneous
transmission of an infinite sized packet sequence. While a burst
function can be approximated by a sufficiently large back-to-back
packet train, a high-volume transmission of probes consumes
network resources and interferes with other packet traffic. In
fact, the service curve of a burst function (or its
approximation) may cause some networks that operate in a min-plus
linear regime to become non-linear. The observation that large
packet trains can lead to unreliable estimates has been noted
in the literature~\cite{dovrolis:packetdispersion}.

In the next section, we present derivations for three
 bandwidth estimation methods in min-plus linear systems. We are able
to relate two of these methods to previously proposed probing schemes.
We will later discuss how these schemes can be applied to certain non-linear systems.

We conclude this section with remarks on some general aspects of
probing schemes and their representations in min-plus
linear system theory.

\noindent$\bullet$ {\bf Timestamps and asynchrony of clocks:}
When clocks at the sender and receiver of a probing packet  are
perfectly synchronized, and the sender includes the transmission time
into each probing packet, the receiver can
accurately construct the functions $A^p$ and $D^p$.
In practice, however, clocks are not synchronized.
When clocks have a fixed offset (but no drift), the arrival
function $A^p$ can be viewed as being
time-shifted by an unknown offset $T$. In the min-plus algebra a
time-shift can be expressed by a convolution, i.e.,
 $A^p (t - T) = A^p \conv \delta_{T} (t)$
where $\delta_{T}(t)=\delta(t-T)$. Here, the convolution of arrival
function and service curve becomes $(A^p \conv \delta_T) \conv \Smin$,
which due to associativity and commutativity of the convolution
operation, can be rewritten as $A^p \conv (\Smin \conv \delta_T)$.
Hence, when the offset is fixed but unknown, even an ideal probing
scheme can only compute a service curve that is a time-shifted
version of the actual service curve of the network. Drifting clocks
make the problem harder. Many bandwidth
estimation schemes circumvent the problem of asynchronous clocks by
 returning probes to the sender ~\cite{akella:bfind,cprobe},
or by  only recording time differences of incoming
probes~\cite{hu:IGI,jain:pathload,topp,ribeiro:pathchirp,strauss:spruce}. A moment's
consideration shows that knowledge of the differences between the
transmission and arrival of probing packets has the same
limitations as dealing with an unknown clock offset $T$ between the sender
and receiver of probing packets.

\noindent$\bullet$  {\bf Losses:} Probe packets that are dropped in the network can be thought of as incurring an infinite delay.  The presentation of arrival and departure functions in Fig.~\ref{fig-probe} is not well suited for accommodating packet losses. An alternative presentation, which expresses arrival and departure times of probe packets (on the y-axis) as a function of the sequence numbers (on the x-axis) can deal with packet losses more elegantly, but may appear less intuitive. Such a description of traffic with flipped axes leads to a dual representation of the network calculus which is based on a max-plus algebra \cite{chang:performanceguarantees,leboudec:networkcalculus}.

\noindent$\bullet$  {\bf Packet pairs:} The arrival and departure
functions of a packet pair have each only three points, i.e., the
origin and the two timestamps related to the packet pair. If it can
be assumed that the service curve has a certain shape, e.g., a
rate-latency curve $S(t) =r\cdot(t-d)_+$, the service curve can be
recovered. In the absence of such an assumption, packet pair methods
may not be able to recover more complex service curves.
This is reflected in observations that bandwidth estimates from packet pairs tend to be less reliable compared to packet trains if cross-traffic is bursty \cite{jain:bandwidthestimationpitfalls,loguinov07}.

\section{Min-Plus Theory of Network Probing Methods}
\label{sec-inverse}

In this section, we derive bandwidth estimation methods as solutions
to finding an unknown service curve for a min-plus system.
For the derivations, we make a number of idealizing assumptions. First, we consider a fluid
flow view of traffic and service. This assumption can be
relaxed at the cost of additional notation. Unless stated otherwise,
we assume that the network represents a min-plus linear system. This
assumption will be  relaxed in Section~\ref{sec-fifo}. We generally assume
that accurate timestamps for transmission and arrival of probes are
feasible. If  measurements only record time differences
between events or include an unknown clock offset between sender
and receiver, the computed service curves need
to be time shifted by some constant value.

\subsection{Passive Measurements}

We  first try to  answer the question: {\em How much
information about the available bandwidth can be extracted from
passive measurements of traffic?} To provide an answer we first
introduce the deconvolution operator of the min-plus algebra, which
is defined for two functions $f$ and $g$ by
\begin{equation*}
f \deconv g (t) = \sup_{\tau} \{f(t+\tau) - g(\tau)\}.
\end{equation*}
The deconvolution operation is {\em not} an inverse to the
convolution ($g \neq f \deconv (f \conv g)$), however, it
has aspects of such an inverse. This is expressed in the following
duality statement from \cite{leboudec:networkcalculus}, which states
that for functions $f$, $g$ and $h$, the following equivalency
holds:\footnote{We use
shorthand notation $f = g \conv h $ to mean `$f (t) = (g \conv h) (t)$
for all $t \geq 0$'.}
\begin{equation}
\begin{split}
f \le g \conv h \quad \Leftrightarrow \quad
h \ge f \deconv g.
\end{split}
\label{eq:duality}
\end{equation}
We will exploit this property to formulate the  following lemma.

\begin{lemma}
\label{prop:composition}
For  two functions $g$ and $h$,
we have

 $$((h \conv g) \deconv g) \conv g = h
\conv g \ . $$
 \end{lemma}

\begin{proof}
The proof makes two applications of Eq.~(\ref{eq:duality}).
Let us define $\tilde{h} = f \deconv g$ and $f = g\conv h$.
By definition of $\tilde{h}$ we can conclude with Eq.~(\ref{eq:duality}) that  $f \le g \conv
\tilde{h}$.

By definition of $f$, we see from
Eq.~(\ref{eq:duality}) that $h\geq f \deconv g$.
By our definition of $\tilde{h}$, this gives us $h \geq \tilde{h}$. From $h \geq
\tilde{h}$ and $f = g \conv h$ we get  $f \ge g \conv \tilde{h}$.

Combining the two statements about the relationship of
$f$ and $g \conv \tilde{h}$ gives us $f = g \conv \tilde{h}$. Now,
by inserting our definition $\tilde{h} = f \deconv g$,
we obtain $f = g \conv (f \deconv g)$. Inserting our second
definition  $f=g \conv h$ yields $g \conv h  = g \conv ((g \conv h)
\deconv g)$. Reordering the expression using  commutativity of the
min-plus convolution completes the proof.
\end{proof}

The lemma justifies the following passive
measurement scheme. Let us denote the arrival and departure
functions measured from a traffic trace of one or more flows  by
$A^{p}$ and $D^{p}$. By assumption of linearity, we know that
$D^{p} = A^{p} \conv S$ holds, but the shape of $S$  is unknown.
Suppose  we compute a function $\tilde{S}$ from the trace as the
deconvolution of the departures and the arrivals, i.e., we set
\begin{equation}
\tilde{S} = D^{p} \deconv A^{p} \ .
\label{eq:deconv-x0}
\end{equation}
With this, we can derive as follows:
\begin{eqnarray*}
D^p & = & S \conv A^p \\
 & = & ((S \conv A^p) \deconv A^p) \conv A^p \\
 & = & (D^p \deconv A^p) \conv A^p \\
 & = & \tilde{S} \conv A^p
\end{eqnarray*}
Equality in the first line holds because of our assumption of linearity. In the second line we apply Lemma~\ref{prop:composition}. The third line uses again the linearity assumption.
In the fourth line, we insert Eq.~(\ref{eq:deconv-x0}).
We can therefore conclude with Lemma~\ref{prop:composition} that
\begin{equation}
D^{p} = A^{p} \conv \tilde{S} \ . \label{eq:deconv-x1}
\end{equation}
Applying the duality property from Eq.~(\ref{eq:duality}) to
$D^p = A^p \conv S$, we obtain $S \geq D^p \deconv A^p$. Then,
with Eq.~(\ref{eq:deconv-x0}) we have
$$
\tilde{S} \leq S \ .
$$
Hence, by deconvolving $D^p$ and $A^p$ as in Eq.~(\ref{eq:deconv-x0}), the result $\tilde{S}$ is a lower service curve, i.e.,  for all
pairs of arrival and departure functions $(A, D)$,
we have $D \geq A \conv \tilde{S}$. Since, from
Eq.~(\ref{eq:deconv-x1}),  $\tilde{S}$ can  completely reconstruct
the departure function from the arrival function, we can conclude
that  $\tilde{S}$ is the best possible estimate of the actual
service curve that can be justified from measurements of $A^{p}$ and
$D^{p}$, in the sense that it extracts the most information from the
measurements.
Since the above deconvolution computes
the largest available bandwidth that can be
justified from a given traffic trace, the described method will perform no
worse than any existing MBAC method from the MBAC literature
\cite{cetinkaya:egressadmissioncontrol}.

The main drawback of this method is that it can only be applied to
linear networks. For networks that do not satisfy min-plus
linearity, i.e., that can only be described by a lower service curve
($D \geq A \conv \underline{S}$) or upper service curve
($D \leq A \conv \overline{S}$),  $\tilde{S}$  only computes a (not
useful) lower bound for an upper service curve $\overline{S}$.
As another remark, note that Lemma~\ref{prop:composition} does not
help us with designing a probing scheme,
since it does not tell us how to select the traffic $A^{p}$ for the
network probes.

For illustration of the passive measurement scheme, we now present two numerical examples.

{\bf Example 1: Sensitivity of Passive Measurements. }
We study the the sensitivity of the passive measurement method with
respect to the burstiness of the trace, the fraction of available
bandwidth that is utilized by the flows, and the length of the
measurement period.
We consider an  idealized fluid flow traffic at a min-plus linear system, which is
governed by a service curve
\[ S(t) = (b+rt) \ast (R[t-T]^+) \ .
\]

The system represents a
network where the input is regulated with a leaky-bucket with
parameters $b$ and $r$, and the service is described by a
latency-rate service curve with delay $T$ and rate $R$. We set
$b=0.75$~Mb, $r=25$~Mbps, $R=100$~Mbps, and $T=10$~ms.

\begin{table}[t]
\centering
\caption{Example 1: Parameters of On-Off sources.}
\label{tab:onoff}

\begin{tabular}{|c||r|r|r|}
 \multicolumn{4}{c}{\small (a) \sc high load}  \\
\hline  Burstiness  & high & med & low \\
\hline Number of sources & 1 & 5 & 25  \\
Source peak rate [Mbps] & 200 & 40 & 8  \\
Total average rate [Mbps] &  20 & 20 & 20  \\
\hline
\end{tabular}

\begin{tabular}{|c||r|r|r|}
  \multicolumn{4}{c}{\small (b) \sc  low load} \\
\hline  Burstiness  & high & med & low \\
\hline Number of sources &  1 & 5 & 25 \\
Source peak rate [Mbps]  & 200 & 40 & 8 \\
Total average rate [Mbps] &   10 & 10 & 10 \\
\hline
\end{tabular}

\end{table}

As traffic trace, we use an arrival sample path that
represents the aggregate arrivals from a set of statistically
independent On-Off traffic sources. In the {\em On state}, each source
generates traffic at  a given peak rate. In the {\em Off state},
no data is generated. In each time slot of duration one  millisecond, a source switches from the
{\em On state} to the {\em Off state} with probability $p$, and from the {\em Off state} to the {\em On state} with
with probability $q$.

The parameters are depicted in Table~\ref{tab:onoff}. In the {\em
high load} setting,
we set $p=0.09$ and $q= 0.01$, resulting in a total arrival rate
of $20$~Mbps. In {\em low load},
we set $p=0.19$ and $q= 0.01$, which leads to an
average total traffic rate of $10$~Mbps.
We control the burstiness of
the traffic by increasing the number of flows, and accordingly
decrease the peak rate of each flow. Due to statistical
multiplexing, an aggregate of multiple On-Off sources is less bursty
than a single flow with the same peak and average rate. In our plots
burstiness levels of {\em high}, {\em medium}, and {\em low}
correspond to a trace with  1,~5, and 25~sources.

In Fig.~\ref{fig:onoffa}-\ref{fig:onoffd} we show the estimates of the
lower service curves $\tilde{S}$  obtained with the deconvolution described above,
and compare them to the actual service curve $S$, indicated as a
thick (red) line in each graph. The length of the measurement is
taken  over  1~second (plots on the left), 10~seconds (plots on the
right). In all plots, we see that burstier traffic leads to better
 estimates of the service curve.  This is expected since we know that the
burstiest traffic, i.e., a burst impulse, can perfectly recover $S$
(see Section~\ref{sec-probetheory}). For the same reason, the
estimates improve when the traffic trace has a higher utilization of
the available bandwidth. Observe that all estimates improve with
increasing length of the evaluation period. This follows from the
definition of the supremum in the min-plus deconvolution operation.

\begin{figure}

\centering
\subfigure[High load, after 1
second.]{\epsfig{file=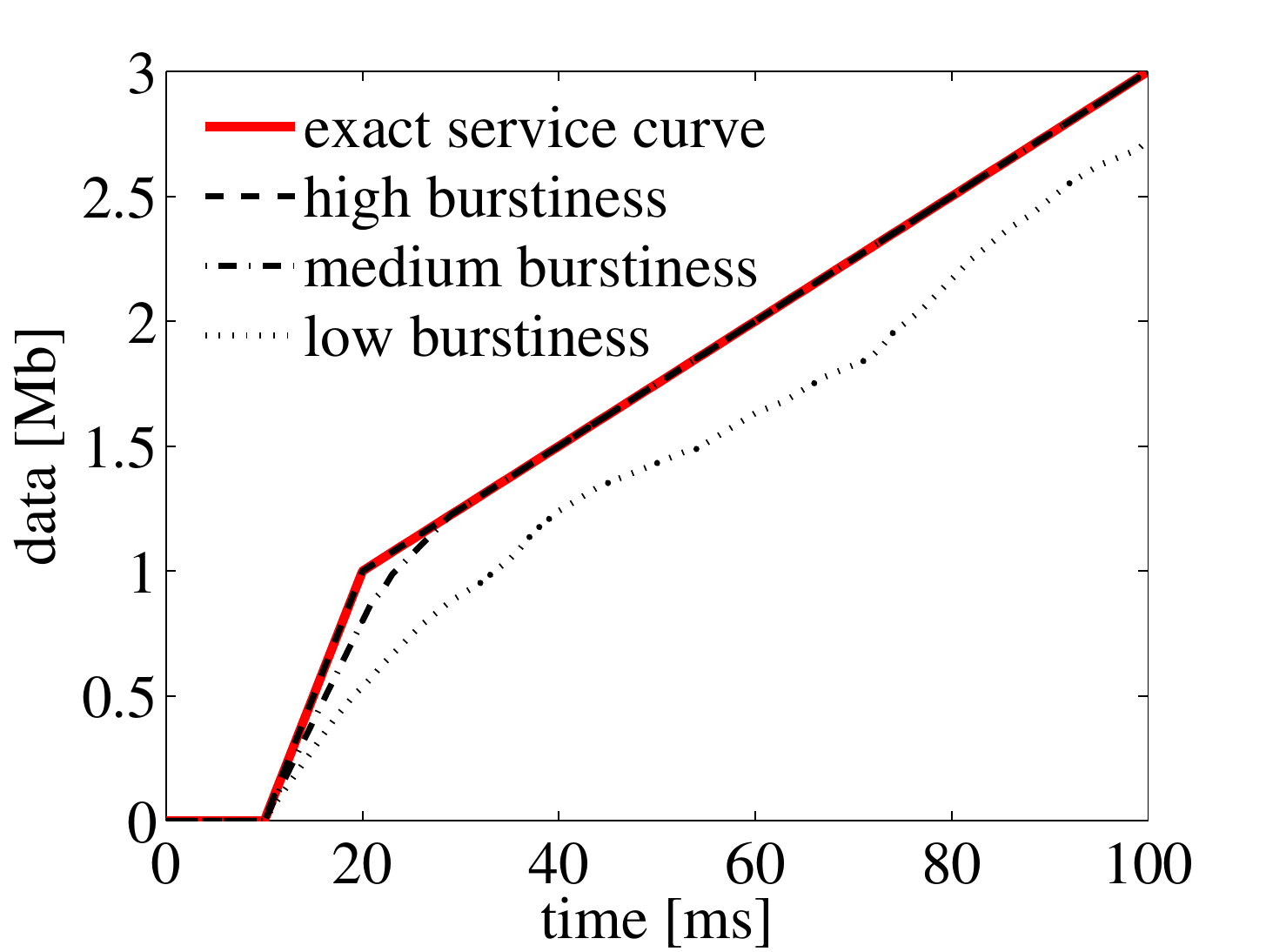, width=0.49\linewidth}\label{fig:onoffa}}
\subfigure[High load, after 10
seconds.]{\epsfig{file=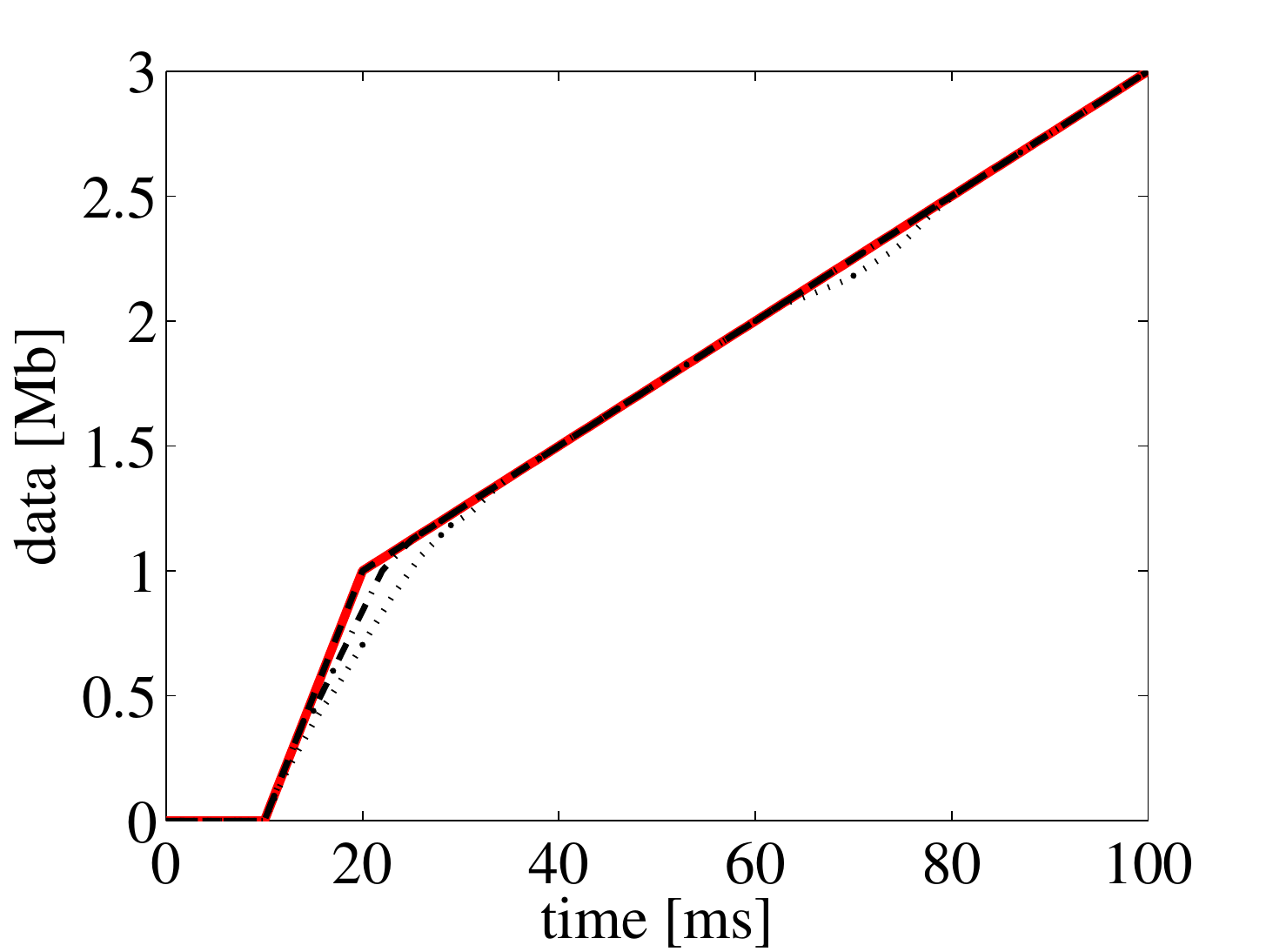, width=0.49\linewidth}\label{fig:onoffb}}
\subfigure[Low load, after 1
second.]{\epsfig{file=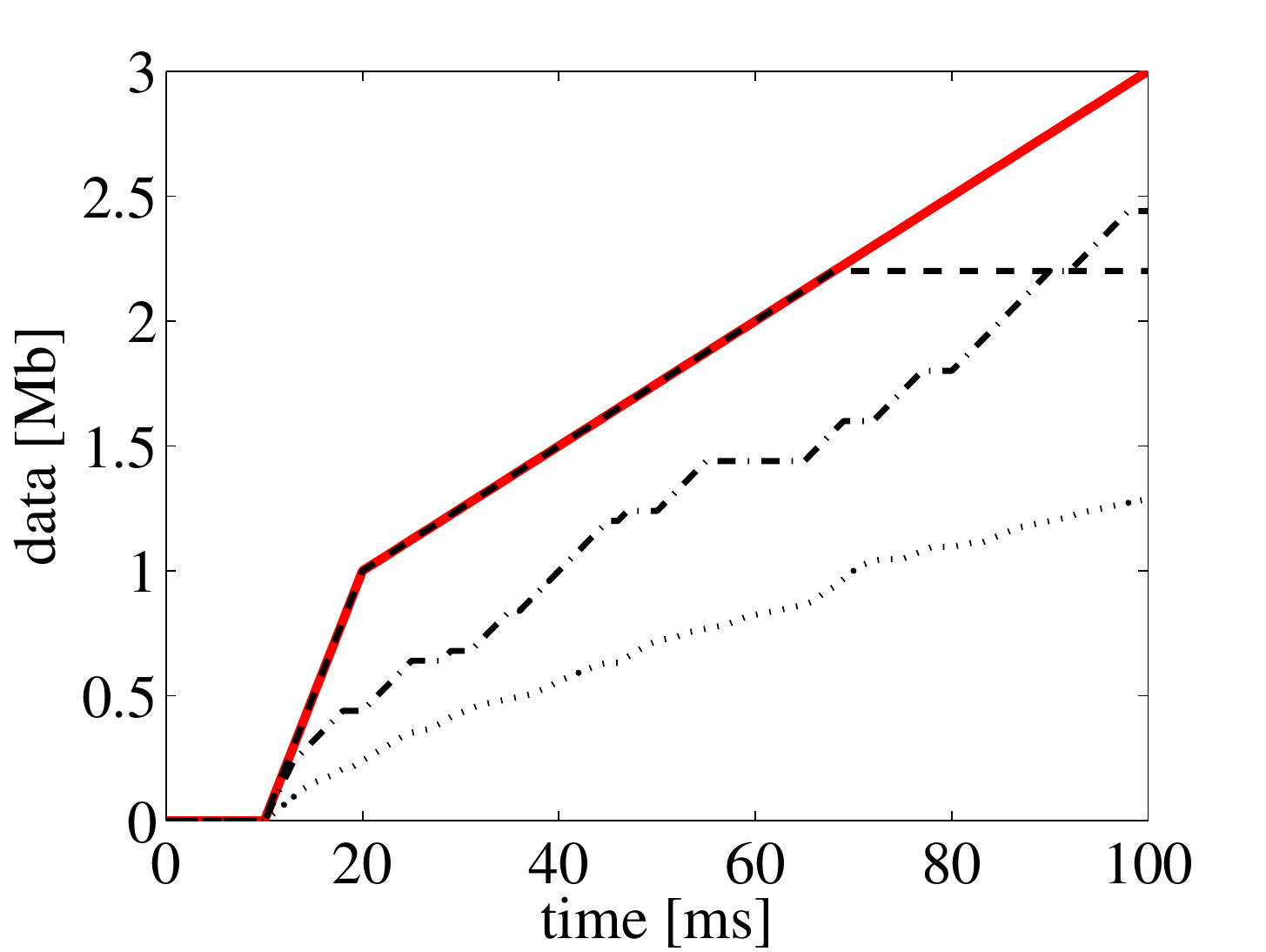, width=0.49\linewidth}\label{fig:onoffc}}
\subfigure[Low load, after 10
seconds.]{\epsfig{file=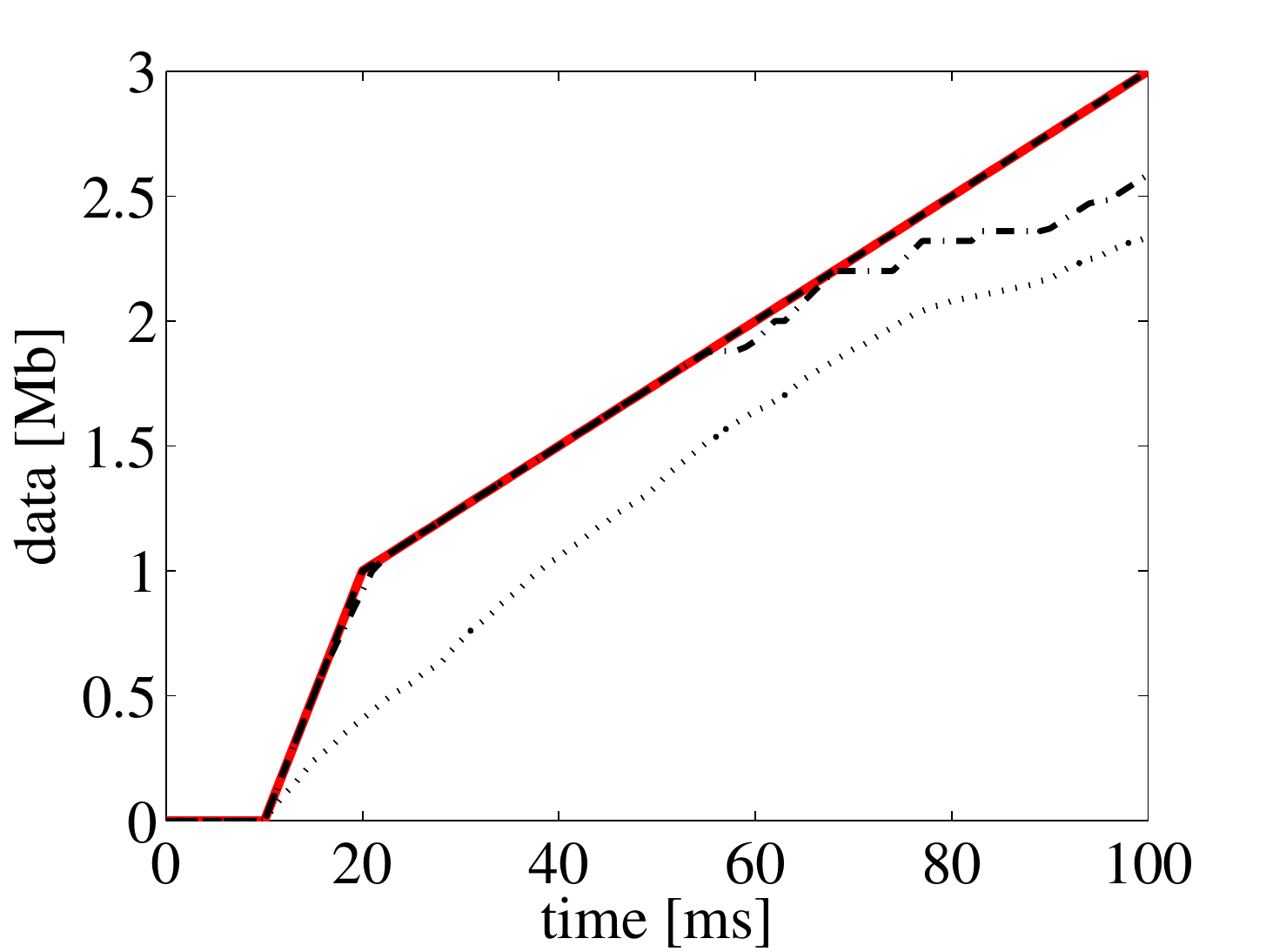, width=0.49\linewidth}\label{fig:onoffd}}

\caption{Example 1: Passive measurement of multiplexed On-Off traffic.\label{fig:video}}
\end{figure}

{\bf Example 2: The Dilemma of Passive Measurements. }
To illustrate the limitations of passive measurements for
bandwidth estimation, we now present as a second example an ns-2 simulation \cite{ns2} of
measurements at a single node with capacity~$C$. There is a  propagation delay of 10~ms at the ingress link and a 10~ms delay at the egress link. The packet scheduling algorithm is
either FIFO or Deficit Round Robin (DRR). DRR approximates a fair
queuing discipline, which can distribute capacity equally among cross and probe traffic. The cross traffic at this link
consists of CBR traffic which is transmitted in 800 byte packets.
The rate of cross traffic is set to half the link capacity.
The traffic source for passive measurements is a small
segment of a high-bandwidth variable bit rate video trace
\cite{Reissleintrace} with an average rate of 17.1~Mbps and a peak
rate of 154~Mbps. (We have used  two seconds of the video
trace entitled {\em From Mars to China}.) We evaluate the bandwidth estimation, when the link capacity is set to
$C=70,~50,~30$~Mbps. The resulting service curves
are shown in Fig.~\ref{fig:video}. In each
figure, the exact service curve (red line) is  a latency rate service curve with delay~20~msec and
rate $C/2$.  The computed estimates are indicated by a dashed  line for FIFO
and a solid line for DRR scheduling. For $C=70$~Mbps, the
available bandwidth is clearly underestimated. The estimates
improve for $C=50$~Mbps, where the video trace accounts for a
larger fraction of the unused bandwidth. For $C=30$~Mbps, the
available bandwidth is estimated with high accuracy for the DRR
link, but overestimated for the FIFO link. The overly optimistic  estimates at a  FIFO link occur when the variable bit rate of the video traffic overloads the link,
thereby preempting cross traffic. An explanation for this outcome
is given in Section~\ref{sec-fifo}, where we discuss non-linearities observed in overloaded FIFO systems.  The video trace example indicates a fundamental dilemma
with passive measurements. On the one hand, if the traffic
intensity of the measured trace is too low, the trace does not
extract enough information from the network. On the other hand,
if the traffic intensity is too high, the traffic trace may
preempt other traffic, thus leading to inaccurate estimates.

\begin{figure}
\centering
\subfigure[$C=70$~Mbps.]{\epsfig{file=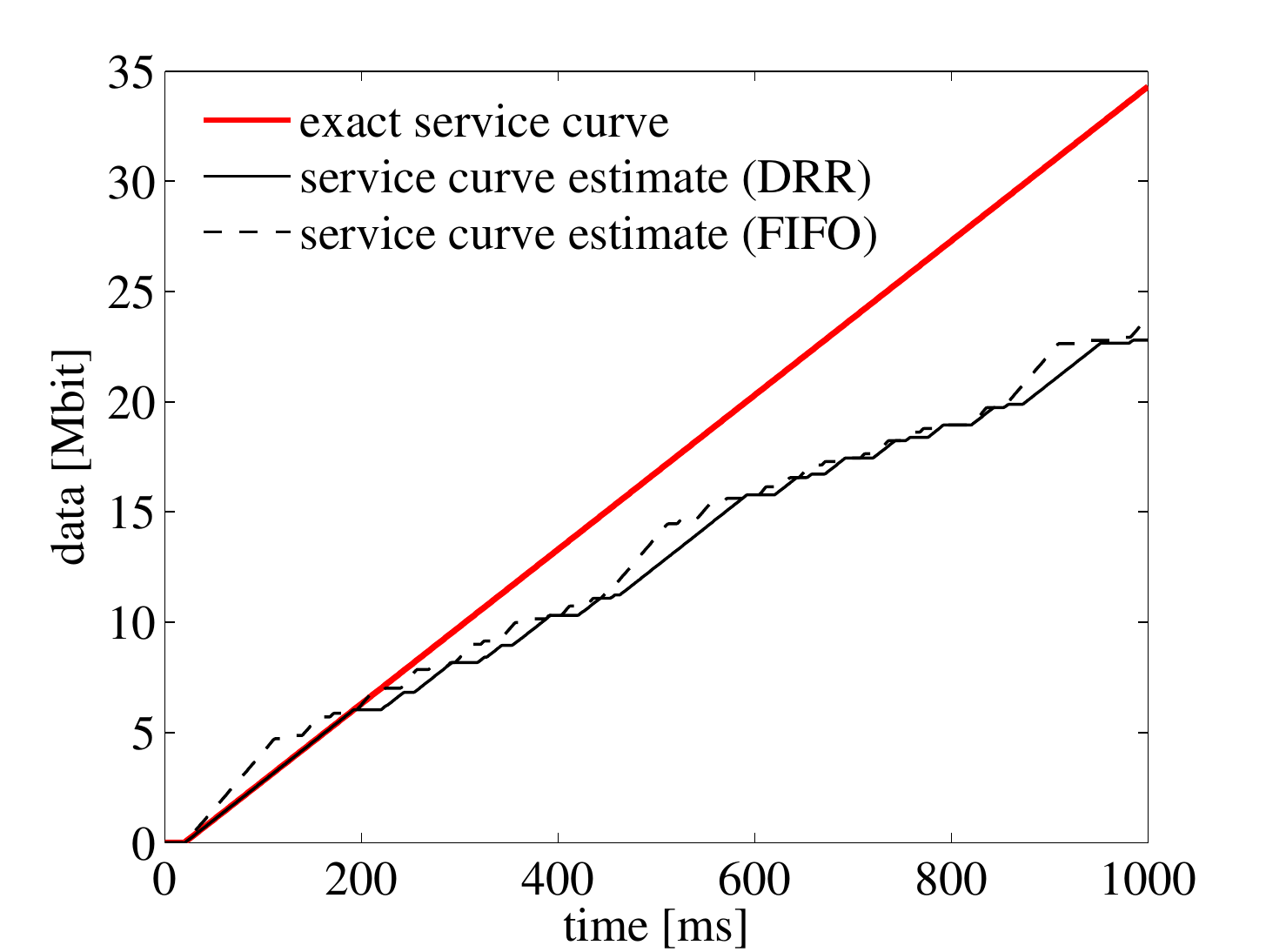,
width=0.49 \linewidth}\label{fig:videoa}}
\subfigure[$C=50$~Mbps.]{\epsfig{file=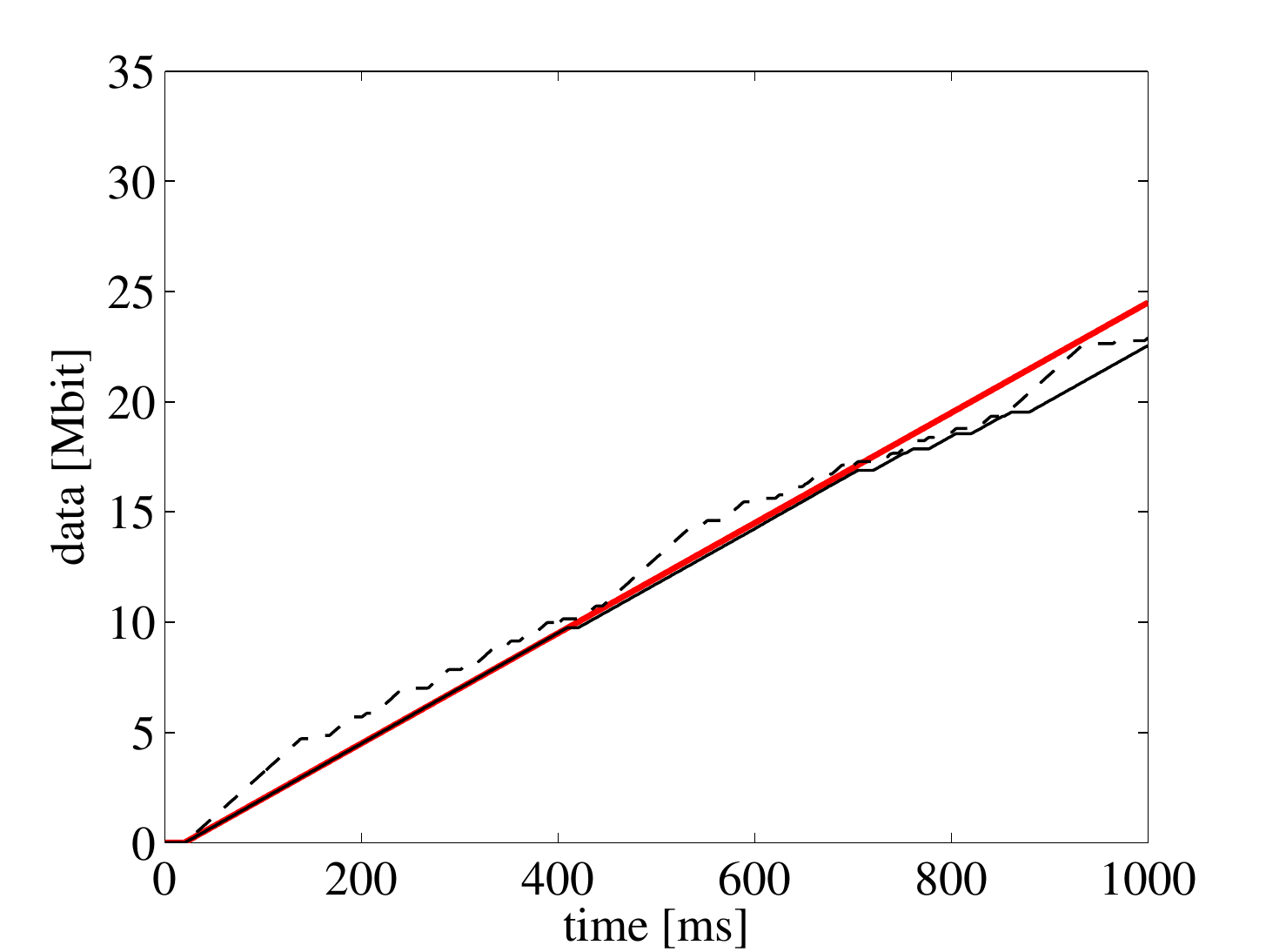,width=0.49
\linewidth}\label{fig:videob}}
\subfigure[$C=30$~Mbps.]{\epsfig{file=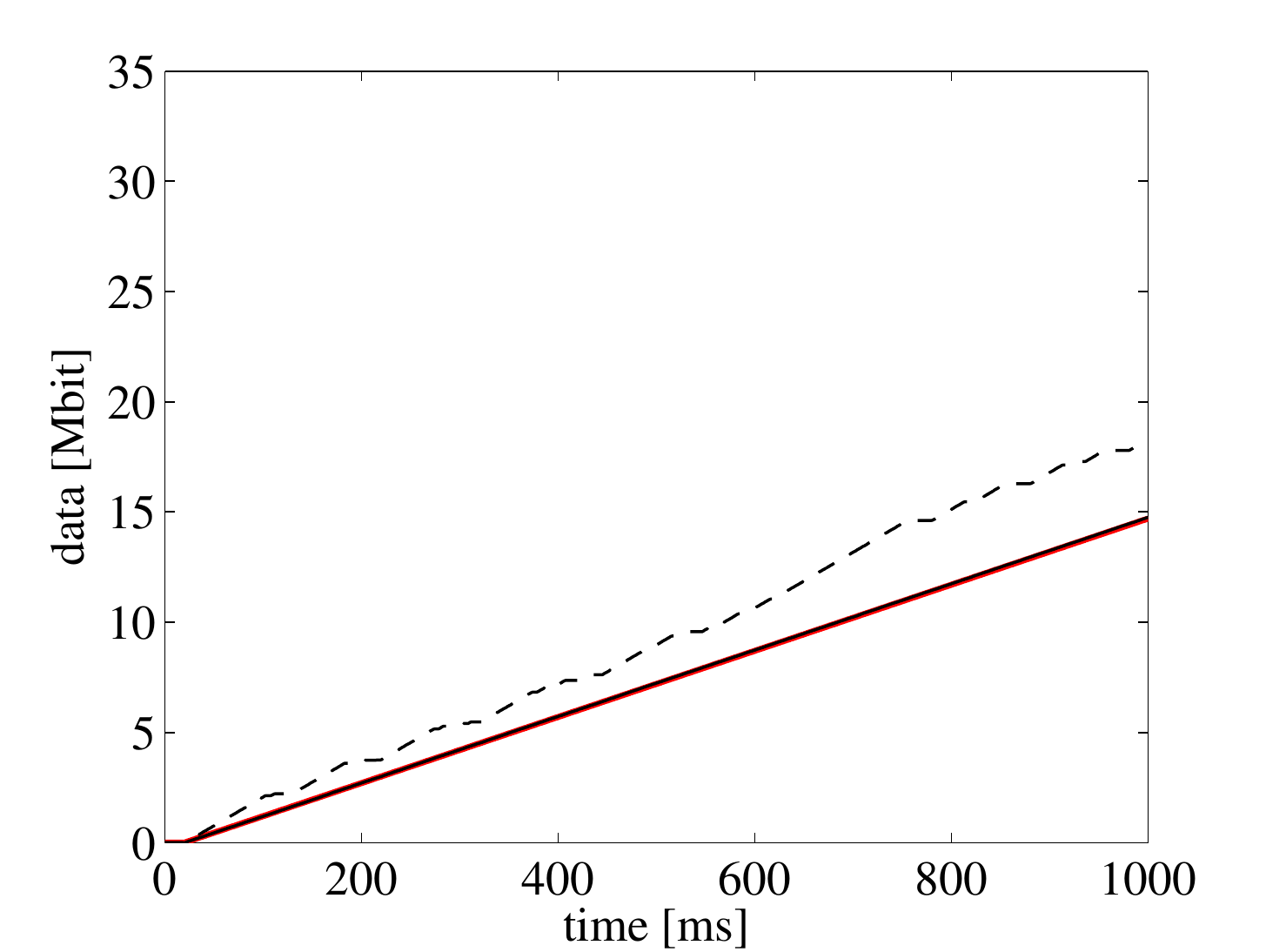,width=0.49 \linewidth}\label{fig:videoc}}
\caption{Example 2: Passive measurement simulation with a video source.}
\end{figure}

\subsection{Rate Scanning}
\label{subsec:ratescan}
We now consider an active probing scheme that transmits packet trains at a constant rate, but varies the rate of subsequent trains,  e.g., such as
{\em pathload}~\cite{jain:slops,jain:pathload}.
We  provide a justification for this approach, which we refer to as {\em rate scanning}, using min-plus system theory.

Given  arrival and departure functions $A$ and $D$, using the
earlier definition of backlog, the maximum backlog
can be computed as
\begin{equation*}
B_{max} = \sup_{t} \{A(t) - D(t) \}.
\end{equation*}
If the arrivals are a constant rate  function, that is, $A(t) = r t, $
and the network satisfies min-plus linearity,
we can write $B_{max}$ as a function of $r$ as follows:
\begin{equation*}
\begin{split}
B_{max}(r) &= \sup_{t} \{rt - \inf_{\tau} \{r \tau + S(t-\tau) \}\} \\
&= \sup_{t} \{ \sup_{\tau} \{r (t -\tau) - S(t-\tau) \}\} \\
&= \sup_{t} \{rt - S(t) \}.
\end{split}
\end{equation*}
The first line uses that  output in min-plus linear systems can be
characterized by  $D= A \conv S$. The second line moves the infimum in
front of the substraction, where it
becomes a supremum. The third line is simply a substitution.

Recalling the definition of the Legendre transform from
Subsection~\ref{subsec-legendre}, the right hand side
of the last equation can be written as the Legendre transform of $S$, that is,
$B_{max} (r) = \Legendre_S (r)$. This relation has been observed
in~\cite{Cruz91,fidler:legendre,naudts:trafficparametermeasurement}.
We now take a further step by applying the relation in the reverse transform.
Due to Eq.~(\ref{eq:leg-convex}), we have for convex service curves
$S$ that
\begin{equation*}
S(t) = \Legendre ( \Legendre_S ) (t) = \Legendre_{B_{max}} (t) = \sup_r \{rt
- B_{max}(r) \} \ .
\end{equation*}
Thus, every convex
service curve can be completely recovered by measurements of the maximum
backlog $B_{max}$.  For service curves that are not convex
one recovers, using Eq.~(\ref{eq:leg-nonconvex}), a lower bound
for the service curve. The interpretation of rate scanning is
that each constant bit rate stream with rate $r$ reveals one point
$B_{max}(r)$ of the service curve in the Legendre domain
$\Legendre_S (r)$.
If we specify a {\em rate
increment}, which sets the rate increase between packet
trains and   a {\em  rate limit}, which sets the maximum rate at
which the network is scanned, we realize a rate scanning method
that computes a service curve consisting of piecewise linear segments. The
choice of the rate increment determines the length of the segments,
and, in this way, the accuracy of the computed service
curve. We note that
rate scanning is capable of tracking a convex service
curve up to a time where the derivative of the service
curve reaches the rate limit. The higher the maximum rate, the more
information about the service curve is recovered. The number of packets in a packet train must be large  enough so that the maximum backlog can be  accurately measured.

A criterion for
picking the rate limit suggested by our derivations is to stop
rate scanning when increasing the scanning rate does not yield an
improvement of the service curve. This criterion, however, may fail when the
underlying network is not min-plus linear.
The rate scanning method {\em pathload} ~\cite{jain:slops,jain:pathload} uses an iterative procedure which
varies the rate $r$ of consecutive packet trains until measured
delays indicate an increasing trend. In
Section~\ref{sec-fifo} we will find that similar criteria can
be justified to determine a rate limit in  a non-linear system.

In Fig.~\ref{fig:ratescanninga} we present an example of the rate
scanning approach for a fluid-flow service curve with a quadratic
form $S(t) = 0.4 t^2$. In the example, rate scanning  is performed
at rates $10,~20, \ldots , 80$~Mbps. In Fig.~\ref{fig:ratescanninga},
we plot the maximum backlog observed for each scanning rate. The
function $B_{max}(r)$ is constructed by connecting the measured data
points by lines. For rates $r$ exceeding the rate limit we can set  $B_{max} (r) = \infty$ to obtain a conservative
 Legendre transform for all rate values.
In Fig.~\ref{fig:ratescanningb}, we show the
service curves that are obtained with different rate limits. The
higher the rate limit, the more accurate the results.
Decreasing the increment of the rate will improve the accuracy of
the service curve.
We point out that  both the backlog plot in
Fig.~\ref{fig:ratescanninga} and the service curves in
Fig.~\ref{fig:ratescanningb} consist of linear segments.

\begin{figure}
\centering
\subfigure[Maximum Backlog $B_{max} (r)$.
]{\epsfig{file=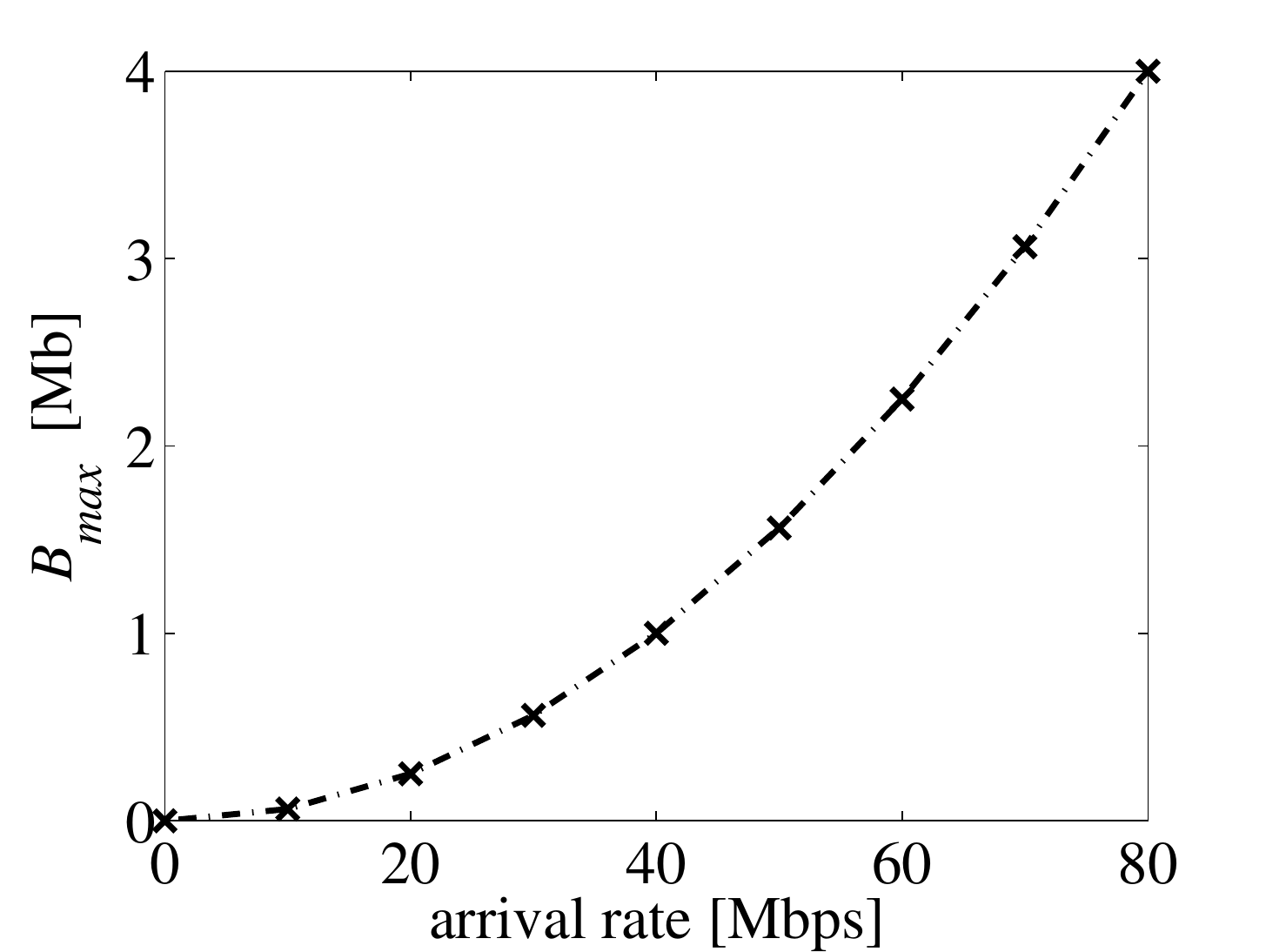, width=0.49
\linewidth}\label{fig:ratescanninga}}
\subfigure[Rate scanning results with different rate limits.]
{\epsfig{file=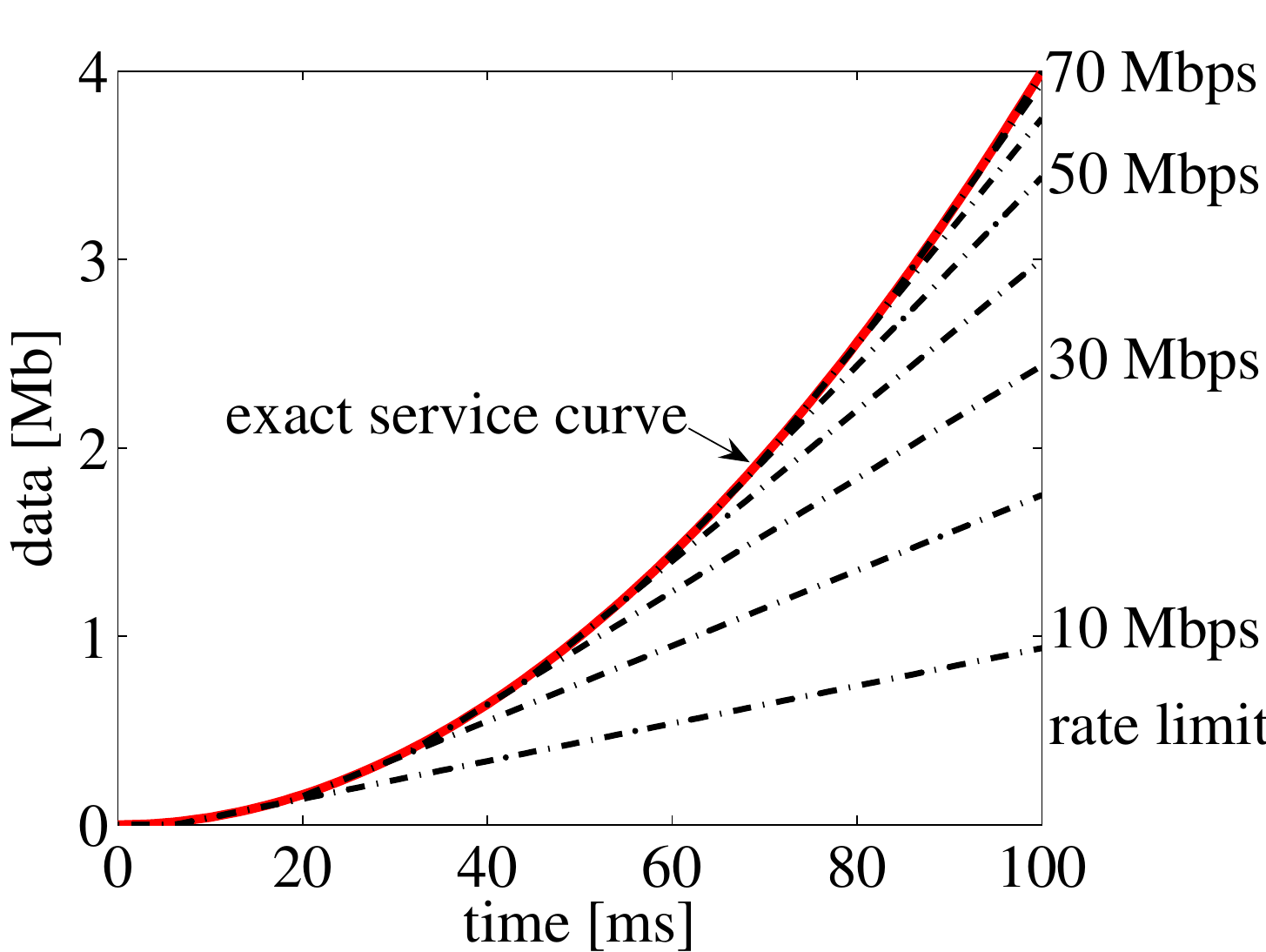, width=0.49
\linewidth}\label{fig:ratescanningb}}
\caption{Service curve estimation with rate scanning.\label{fig:ratescanning}}
\end{figure}

\begin{figure}
\centering
\subfigure[Rate Chirps.] {\epsfig{file=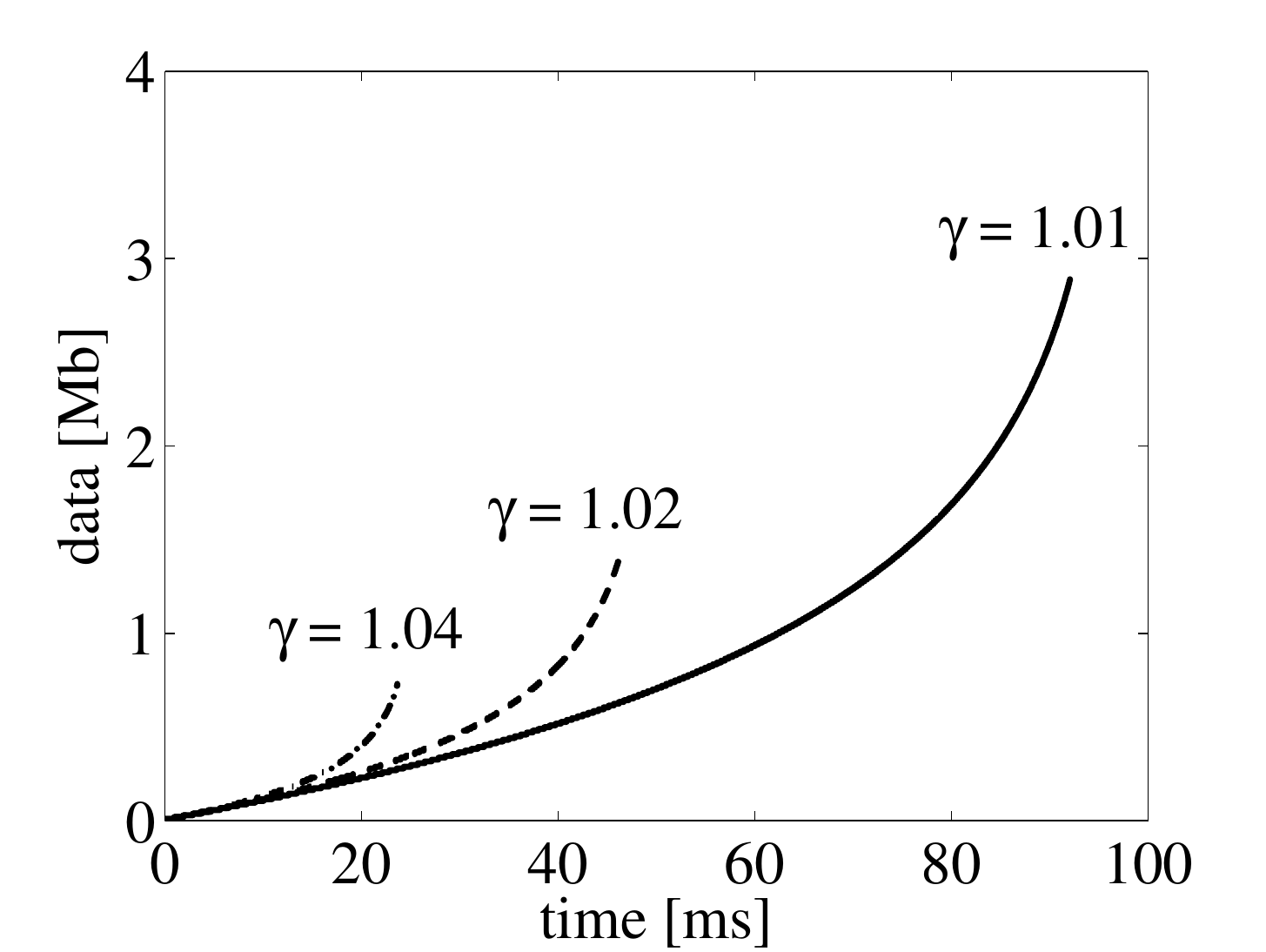,
width=0.49 \linewidth}\label{fig:ratechirpa}}
\subfigure[Rate chirp results with different spread factors.]
{\epsfig{file=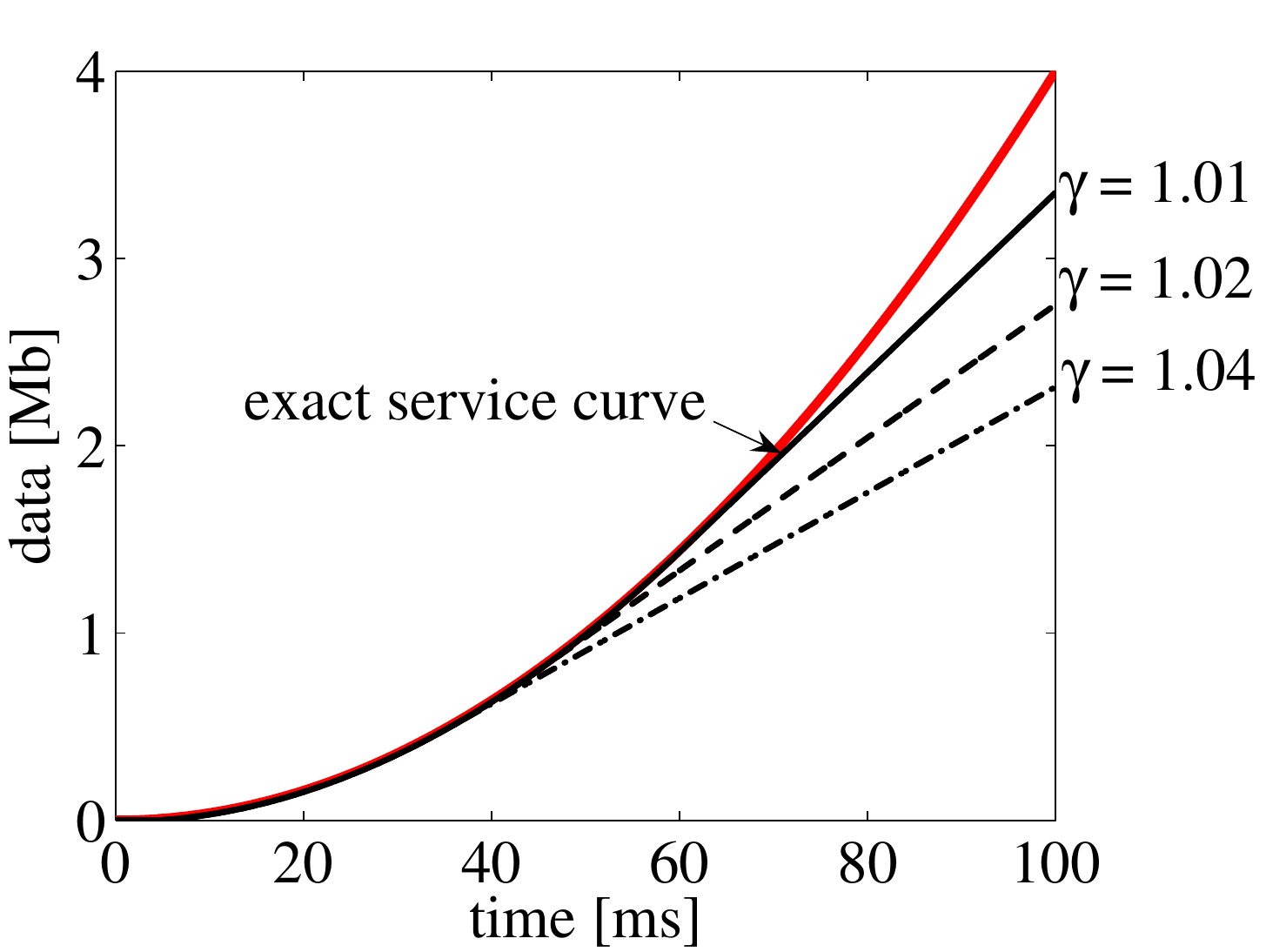, width=0.49
\linewidth}\label{fig:ratechirpb}}
 \caption{Service curve estimation with rate
chirps.\label{fig:ratechirp}}
\end{figure}

\subsection{Rate Chirps}
\label{subsec:ratechirp}

The need of rate scanning  to measure a possibly large number of
packet trains has motivated the {\em pathchirp} method \cite{ribeiro:pathchirp}, where  available bandwidth estimates are based on the measurement of  a single packet train, with a geometrically decreasing
inter-packet spacing.  The approach takes inspiration from chirp
signals in signal processing, which are signals
whose frequencies change with time. We refer to this approach as {\em
rate chirp}, since the decreased gap between packets corresponds to
an increase of the transmission rate.
We will show that a rate chirp scheme can be justified in  min-plus system theory using properties of the
Legendre transform.

Suppose we have a lower service curve $\Smin$ satisfying
$D \geq A \conv \Smin$ for all pairs $(A,D)$. Taking
the Legendre transform we
obtain with the order reversing property of Eq.~(\ref{eq:leg-rev}) and
with Eq.~(\ref{eq:leg-add}),   that
\begin{equation*}
\Legendre_D \le  \Legendre_{A \conv \Smin}
=   \Legendre_A + \Legendre_{\Smin} \label{eq:upperlegendrefirsttransform} \ .
\end{equation*}
We can re-write this  as
\[
\Legendre_{\Smin} \ge \Legendre_D - \Legendre_A \ ,
\]
as long as the difference
$\Legendre_D (r) - \Legendre_A(r)$ is defined for all $r$.
A sufficient condition is that $\Legendre_A (r) < \infty$,
since it prevents  both transforms $\Legendre_D$ and
$\Legendre_A$ from becoming infinite at the same value of $r$.
Another application of Eq.~(\ref{eq:leg-rev}) yields
\[
\Legendre(\Legendre_{\Smin}) \! \leq \!  \Legendre(\Legendre_D - \Legendre_A) \ .
\]
If the system is min-plus linear, that is, $D = A \conv S$, we
get,
\[
\Legendre(\Legendre_S) \! = \!  \Legendre(\Legendre_D - \Legendre_A) \ .
\]
If $S$ is also convex, then by Eq.~(\ref{eq:leg-convex}),
we  have
$S = \Legendre(\Legendre_D - \Legendre_A)$.

This provides us with a justification for
{\em pathchirp} \cite{ribeiro:pathchirp} as a probing method. If we depict
the transmission of a packet chirp as a fluid flow function, we see
that it grows to an infinite rate, thus, yielding a Legendre
transform that is finite for all rates. By measuring  arrivals and
departures of the chirp, denoted by $A^{chrp}$ and $D^{chrp}$, we
can compute a function $\tilde{S}$ by
\begin{equation}
\tilde{S} (t) = \Legendre(\Legendre_{D^{chrp}} - \Legendre_{A^{chrp}}) (t) \ .
\label{eq-chrp}
\end{equation}
If the network satisfies  $D = A \conv S$ for all arrivals, then
the right hand side of Eq.~(\ref{eq-chrp}) computes $\Legendre(\Legendre_S)$.
With  Eq.~(\ref{eq:leg-nonconvex}), we obtain
$ \tilde{S} \leq S$,
which tells us that  $\tilde{S}$ is a lower
service curve that satisfies $D \geq A \conv \tilde{S}$ for any
traffic with arrival function $A$ and departure function $D$.
If $S$ is convex we have $\tilde{S} = S$,
and we can recover the service curve exactly.

In practical probing schemes, our fluid flow interpretation where a packet chirp can grow to an infinite rate is idealized, since a rate chirp cannot be transmitted faster than the data rate at the
sender of probe packets.
For a packet chirp that is transmitted
in a time interval $[0, t^A_{max}]$ and
where $D$ is observed over an interval $[0, t^D_{max}]$,
the following adjustment complies with the formal requirements of
our equations:
\begin{eqnarray}
\tilde{A}^{chrp} (t) & = &
\begin{cases}
A^{chrp} \ ,    & \mbox{if } 0 \leq t \leq t^A_{max} \ , \\
\infty \ ,     & \mbox{if } t > t^A_{max}  \ ,
\end{cases}
\nonumber \\
\tilde{D}^{chrp} (t) & = &
\left\{
\begin{array}{l}
D^{chrp}  \ ,    \mbox{if } 0 \leq t \leq t^D_{max} \ , \\
D^{chrp}   (t^D_{max}) + (t - t^D_{max}) \frac{d D^{chrp}}{dt} (t^D_{max}) \ , \\\
\qquad\qquad \mbox{if } t > t^D_{max} \ .
\end{array}
\right.  \nonumber
\end{eqnarray}
The arrival function is simply set to $\infty$ past the
last measurement. The departure function is continued
at a rate that corresponds to its slope at the time of the last
measurement. For convex service curves $S$, the above extensions are conservative.

In Fig.~\ref{fig:ratechirpa} we show several rate chirps for a
network probe. The rate chirp  consists of a step-function which emulates a sequence of probing packets of 1200~bytes. The packets  are transmitted at an increasing rate, starting
at 10~Mbps and growing to 200~Mbps. The rate is increased by
reducing the elapsed time between the transmission of the first bit
of two consecutive packets, by a constant factor $\gamma$, which is
called the spread factor in~\cite{ribeiro:pathchirp}. Larger values
for $\gamma$ lead to shorter chirps that grow faster to the maximum
rate. In Fig.~\ref{fig:ratechirpb}, we show the service curves
computed from the chirps in Fig.~\ref{fig:ratechirpa}. The actual
service curve is  $S(t) = 0.4 t^2$, indicated as a thick (red) line in the
figure. A chirp with a smaller spread factor
$\gamma$, which transmits more packets over longer time interval,
leads to better estimates of the service curve.

\section{Bandwidth estimation in non-linear systems}
\label{sec-fifo}

Extending bandwidth estimation to systems that are not min-plus
linear, i.e., cannot be described by an exact service curve, raises
difficult questions. First, the problem formulation of bandwidth estimation at the beginning of Section~\ref{sec-probetheory} has shown that the problem has the structure of a maximin optimization. Moreover,  in
networks with non-linearities  the network service available to a traffic flow may depend on the traffic transmitted by this flow.  If this is the case, knowledge of the available bandwidth may not help with predicting network behavior.

In this section, we provide solutions for a class of networks that
can be decomposed into disjoint min-plus linear and
non-linear regions. These networks behave like a min-plus linear system at low load, and become non-linear when the traffic rate is increased beyond a threshold.
In such a network, the goal of
bandwidth estimation should be to determine the available bandwidth  of the linear region.
The interpretation is that  the available bandwidth denotes the maximum additional load that the network can carry without
degrading to a non-linear system.
Our work is motivated by studying the available bandwidth at a FIFO link. While we conjecture that most networks can be adequately described by a system that behaves linearly at low loads, the actual scope of this class of networks remains
an open problem.

\begin{figure}
\centering
\subfigure[Cross and probe traffic.] {\epsfig{file=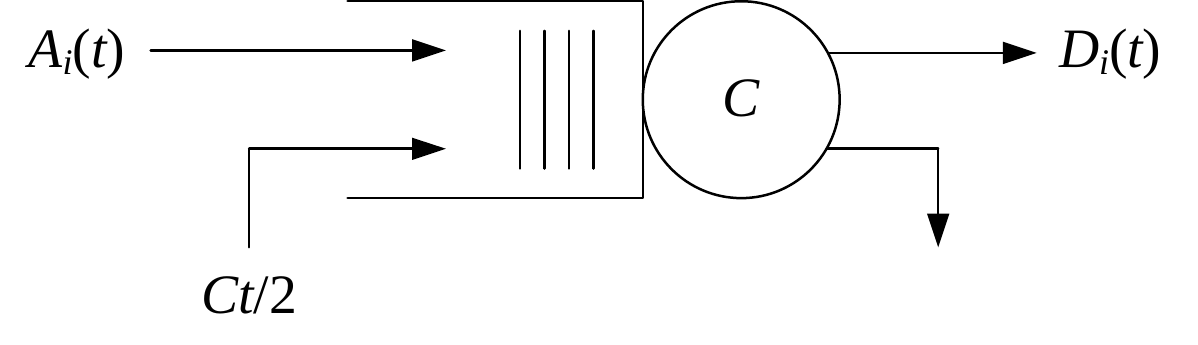,
width=0.49 \linewidth}\label{fig:fifofig}}
\subfigure[Departures of probes.]
{\epsfig{file=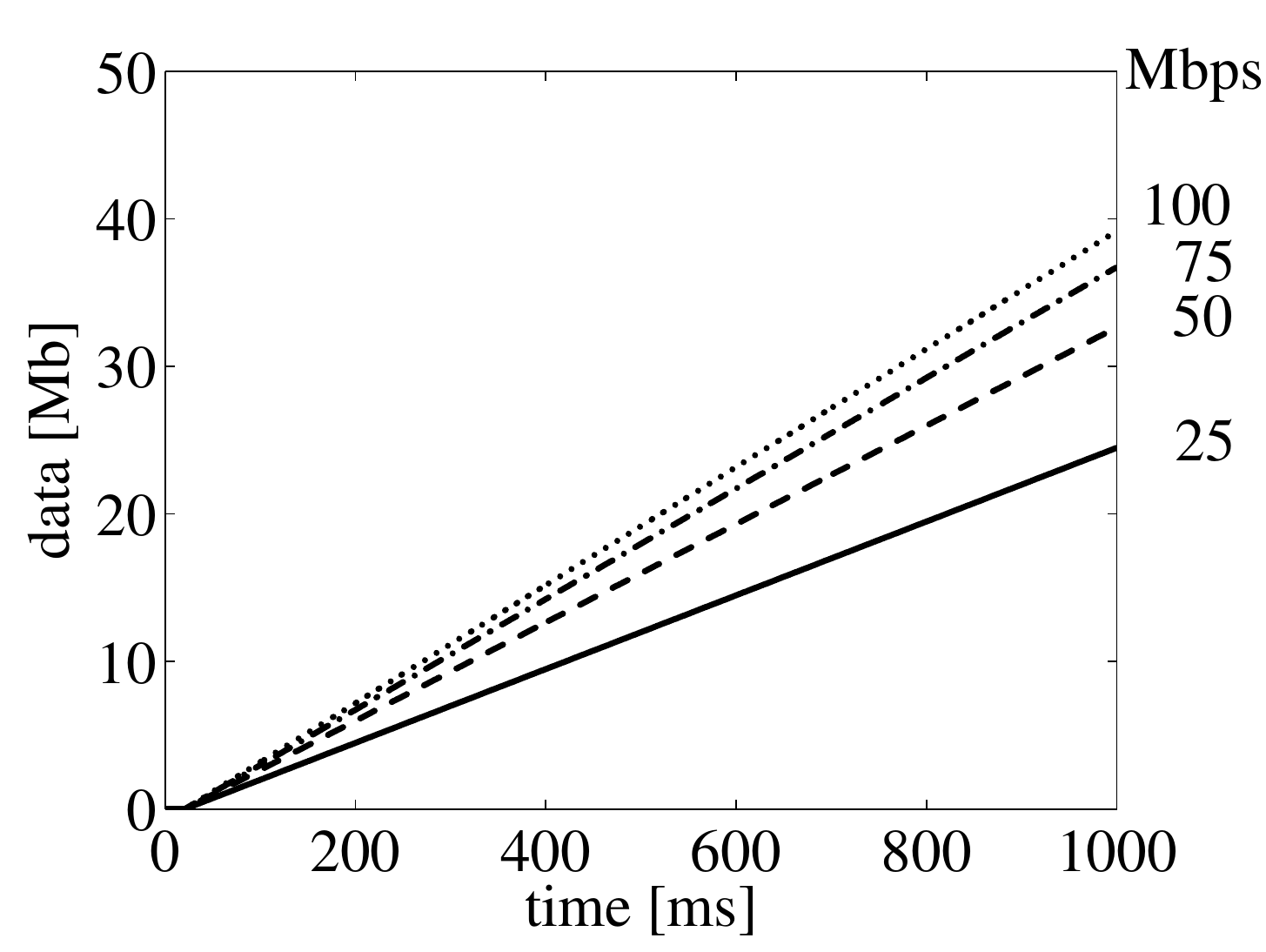, width=0.49
\linewidth}\label{fig:fifo}}
\caption{FIFO system.}
\end{figure}

\subsection{Non-linearity of FIFO systems}

Consider the FIFO system shown in Fig.~\ref{fig:fifofig} with
capacity $C$. Assume that we have constant-bit rate traffic that is transmitted in
800~byte packets. The FIFO queue experiences (cross) traffic at a
rate of $r_c$, and probing traffic is sent according to $A(t) = r
t$. Assuming  a link capacity of $C=50$~Mbps and cross traffic of
$r_c=25$~Mbps, we consider a probing rate of $r=25,50,75,100$~Mbps.
For an {\em ns-2} simulation of this system, Fig.~\ref{fig:fifo} depicts the departure function of the probe
packets for the range of probing rates. As seen previously for passive measurements at a FIFO queue (see
Fig.~\ref{fig:video}), once the probing traffic
exceeds the unused capacity, it preempts cross traffic and results
in an overly optimistic estimate of the available bandwidth.
Empirical observations of  FIFO systems with CBR cross and probe
traffic in~\cite{melander:fcfsprobing} suggested the following departure
function:
\begin{equation}
D(t) = \begin{cases} rt \ , & \text{if } r \le C-r_c \ , \\
\frac{r}{r+r_c} Ct \ , & \text{if } r > C-r_c \ .
\end{cases}
\label{eq:fifo1}
\end{equation}
Thus, if the probing rate is above the threshold $C-r_c$, the
capacity allocated to the probe and cross traffic is proportional to
their respective rates. As a result,  probing traffic gets more
bandwidth when  its rate is increased.

We now  offer a  min-plus system interpretation of bandwidth
estimation for the depicted FIFO scenario. Consider the function
$S_{fifo} (t) = (C- r_c) t$. From the empirical departure
characterization $D$ of a FIFO system from Eq.~(\ref{eq:fifo1}), we
can verify that the following is satisfied for all $t \geq 0$:
\begin{equation}
\begin{array}{l l}
D(t)  & = ( rt ) \conv S_{fifo} \ , \text{if } r \le C-r_c  \\
D(t)  & \geq ( rt ) \conv S_{fifo} \ , \text{if } r > C-r_c  \ .
\end{array}
\label{eq:fifo2}
\end{equation}
Therefore,  $S_{fifo}$ is an exact service curve for $A(t) = rt$
when $r \leq C - r_c$, and  $S_{fifo}$ is a lower service curve when
the arrivals exceed the threshold value. In fact, $S_{fifo}$ is the
largest lower service curve for a FIFO system, and a solution to the
maximization in Section~\ref{sec-probetheory}. Any function larger
than $S_{fifo}$ may not be a lower service curve for rates $r> C-r_c$,
indicating that a FIFO system is not min-plus linear in this range.

These considerations suggest to view a FIFO network as a system that
is min-plus linear at rates $r \le C-r_c$, and crosses into a
non-linear region when the rate exceeds the threshold. The crossing
of these regions coincides with the point where the available
bandwidth $S_{fifo}$ can be observed.

Probing schemes that vary the rate of probe traffic can sometimes be
interpreted in terms of searching for the crossover from a linear to
a non-linear regime. In particular, the rules in {\em pathload} and
{\em pathchirp} to stop measurements when increasing delays are
observed can  be justified in terms of crossing the non-linear
region (at least in a FIFO system), since a probing rate above $C -
r_c$ is the turning point when the buffer of the FIFO system fills up. In the
remainder of this section, we address the problem of locating this
crossover point using systems-theoretic arguments.

\subsection{Stopping Criteria}
\label{subsec:stop}

We address the problem of determining the threshold
probing rate for a system with
disjoint linear and non-linear regions. The threshold probing rate can be interpreted as the maximum rate at which the network can be probed without leaving the linear region.  We refer to a condition that determines the maximum probing rate as a {\em stopping criterion.}

{\em Non-linearity Criterion: } In a min-plus linear system, the service curve is independent of the traffic intensity of the probe traffic. If we have obtained, under assumption of min-plus linearity, a lower service curve $\tilde{S}$  from a measurement probe with functions $(A, D)$, then $\tilde{S}$ must be a lower service curve for any other arbitrary measurement probe  $(A', D')$, that is, $D' (t) \geq A' \conv \tilde{S} (t)$ for  all times $t$. A violation of the inequality indicates that the assumption of linearity used for the computation of $\tilde{S}$ is false.

A simple non-linearity test can be devised for systems where increased traffic does not result in decreased output.
Consider a sequence of probes  $(A_i, D_i)_{i=1,2,\ldots, n}$, where  the traffic intensity of subsequent probes is increased, that is, $A_{i+1} \geq A_i$. By assumption, we also have $D_{i+1} \geq D_i$. Each probe results in an estimate $\tilde{S}_k (t)$.
If there is a $k$ for which  $\tilde{S}_k$ violates linearity for some $i \leq k$, that is, $D_i (t) < A_i \conv \tilde{S}_k (t)$ for  some $t$, then the network is no longer in the linear region for the probe  $(A_k, D_k)$.

The described criterion can be directly applied to a rate  scanning approach with increasing  probing rates where $A_i (t) = r_i t$ with  $r_{i+1} > r_i$.  As an alternative, one could  modify the scanning rate
to perform a search for the maximum scanning rate in the linear region. This  makes the criterion more similar to the scanning pattern in {\em pathload}.

Applying the non-linearity criterion to a rate chirp approach is less straightforward, since there is only a single
arrival function $A^{chirp}$. Generating  multiple arrival functions from a single rate chirp by
truncating the arrival functions merely produces truncated versions of the same service curve. Transmitting multiple rate chirps with different spread factors (see Fig.~\ref{fig:ratechirpa}) makes the criterion applicable, yet, it loses the main advantage of rate chirps of  requiring only a single packet train.
Thus, with only a single packet train, we are unable to justify a stopping criteria from min-plus linear systems theory.

{\em Backlog Convexity  Criterion:} This method is applicable to the rate scanning methods, with probing rates  $r \in [r_1, r_2, \dots, r_n]$ with $r_{i+1} > r_i$. Assume that the maximum backlog measurement is $B_{\max}(r)$ for rate $r$. Recall from
Subsection~\ref{subsec:ratescan}, that a linear system satisfies
$S(t) = \Legendre_{B_{\max}}(t)$ and  $B_{\max}(r) =
\Legendre_S(r)$ holds for all $r$. This motivates a test for
linearity that exploits properties of the Legendre transform
discussed in Subsection~\ref{subsec-legendre}. Under the assumption
of linearity, an estimate of the service is obtained from
\[
\tilde{S} (t) = \Legendre_{B_{\max}}(t) \ ,
\]
and  $B_{\max}(r) = \Legendre_{\tilde{S}}(r)$ holds for all $r$.
Using Eq.~(\ref{eq:leg-convex}) and Eq.~(\ref{eq:leg-nonconvex}), if
there exists an $r$ such that $B_{\max}(r)$ is not convex, i.e.,
\begin{equation*}
B_{\max}(r) \neq \text{conv}_{B_{\max}} (r)
\end{equation*}
for some $r$, we have $B_{\max}(r) \neq \Legendre_{\tilde{S}}(r)$,
and, hence, the hypothesis of a linear system is dismissed.

For systems that are linear at low probing rates and cross into a
non-linear region after a threshold is reached, a convexity test
can be easily devised for  schemes
that incrementally increase the probing rate. After each rate
step, one simply performs a test for equality of $B_{\max}(r)$ and
$\Legendre(\Legendre_{B_{\max}})(r)$. If
$B_{\max}(r) \neq \Legendre(\Legendre_{B_{\max}})(r)$, the
system has reached the non-linear region and the rate
scan is terminated. Otherwise, it is assumed that the system
is still linear, and the probing rate is increased.

To avoid false positives and negatives in the test for equality, we suggest a heuristic that can
account for variability in the measurements. For each value of $r$, we compute the difference
\begin{equation*}
\Delta B(r) = B_{\max}(r) -
\Legendre(\Legendre_{B_{\max}})(r)
\end{equation*}
Note that $\Delta B(r)$ is generally positive since
$B_{\max}(r) \ge \text{conv}_{B_{\max}}(r)$. When the normalized difference $\Delta B(r) /r$  exceeds a threshold value, we assume that
the system is no longer in the linear region.
Outliers that are due to random fluctuations can be  eliminated by median filtering $\Delta B(r)/r$ before applying the threshold test.
In our experiments, we
use a threshold of $\alpha = 4$~ms and perform median filtering.

An additional issue is that, since  packet trains
have a certain length, the maximum backlog  $B_{\max}$
may not be attained by a train. In order to
apply the backlog convexity criterion to finite-length packet trains, it
must be shown that the backlog that is created by fixed length
packet trains also violates convexity once the boundary to the
non-linear region is crossed. As an example, for FIFO systems, we obtain from
Eq.~(\ref{eq:fifo1}) that the maximum backlog  generated by a
packet train of $L$ bits is
\begin{equation*}
B^L_{\max}(r) =
\begin{cases}
0 \; , & \text{ if } r \le C - r_c \\
L \cdot \left( 1 - \frac{C}{r + r_c} \right) \; , & \text{ if } r
> C - r_c.
\end{cases}
\end{equation*}
For all $r > C - r_c$ the second derivative of $B^L_{\max}(r)$ is
negative and thus $B^L_{\max}(r)$ is strictly concave, while it is
convex for $r \le C - r_c$. Thus, the backlog convexity criterion can be
applied for finite packet trains in this case.

\section{Experimental Validation}
\label{sec-emulab}

In this section, we present measurement experiments on an IP network
 that provide an empirical evaluation of the proposed
system theoretic approach to bandwidth estimation.
Specifically, we attempt to provide answers to the following questions:
\begin{itemize}
\item How well does the described min-plus systems theory which
assumes an idealized fluid-flow characterization of traffic and
service translate in a packet based environment?
\item How robust are the available bandwidth methods to changes of the distribution of the cross traffic?
\item How well is a min-plus systems theoretical approach suitable for finding end-to-end estimates
over multiple links?
\end{itemize}

We conduct a series of measurement experiments on the {\em Emulab}
network testbed at the University of Utah \cite{emulab}, where
experiments are run on a cluster of PCs that are interconnected by a switched Ethernet network. Propagation delays are emulated by  PCs
that buffer packets in transmission. Emulab provides a realistic IP network environment, yet it
offers a controlled lab
environment  where traffic and resource availability can be explicitly configured. The ability to precisely control network resources enables us to evaluate how well available bandwidth estimates match the configured availability of network resources.

In our experiments, we take advantage of the fact that
system clocks in the Emulab testbed are synchronized up to 1~ms. According to our discussion in
Section~\ref{sec-probetheory}, if synchronized clocks are not available, then
the service curves computed in this section should be interpreted as being horizontally displaced by an unknown amount.

We have implemented the probing schemes for rate-scanning  (Section~\ref{subsec:ratescan}) and  rate chirps (Section~\ref{subsec:ratechirp}) using the {\em rude-crude} traffic generator \cite{rude-crude}. In addition, in some experiments we include for benchmark comparison the results of measurements using an unmodified version of the {\em pathload} software.

\begin{figure}
\centering \epsfig{file=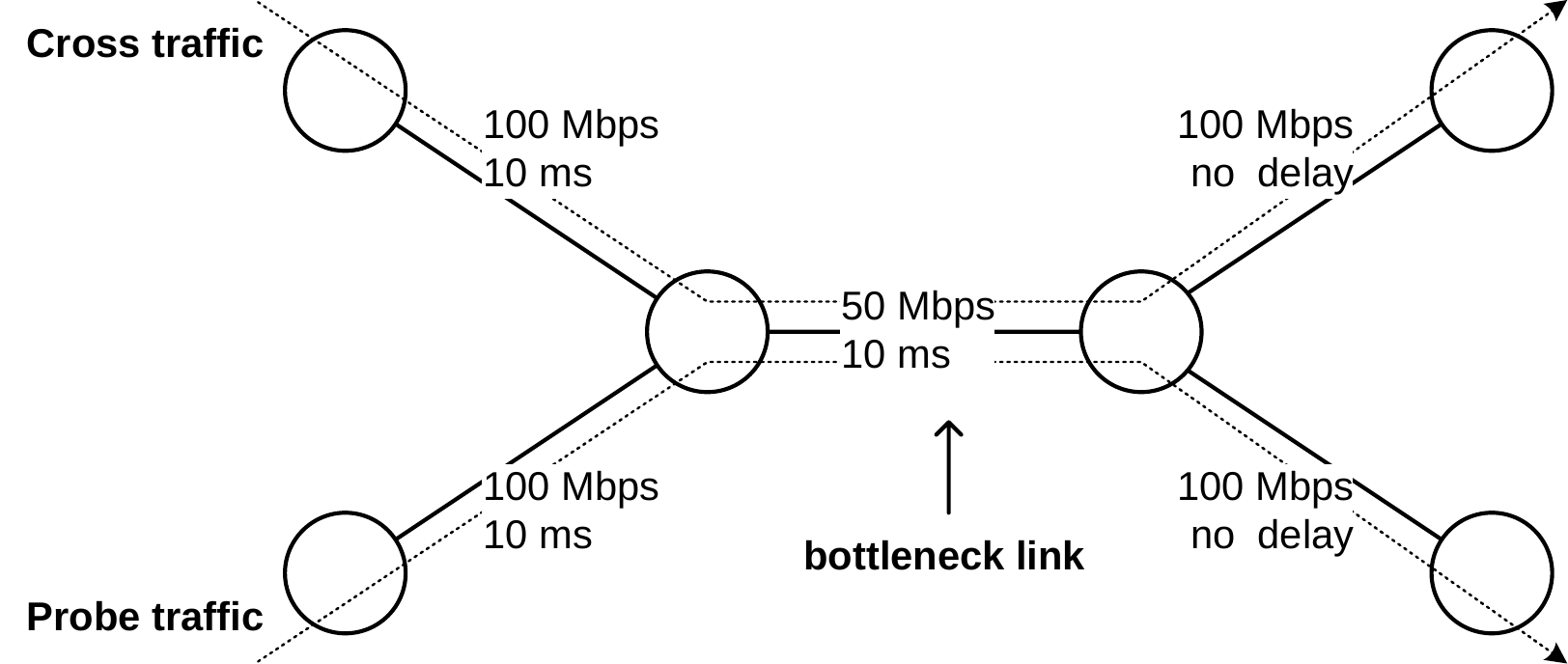, width=0.7 \linewidth}
\caption{Dumbbell topology.} \label{fig:dumbbell}
\end{figure}

We first present measurements on a dumbbell topology as shown in
Fig.~\ref{fig:dumbbell}, where each node is realized by a PC of
the Emulab network. The figure indicates the capacity and the
latency of each link. Packet sizes are set to 800 bytes for cross
traffic and 1472 bytes for probing traffic. The average data rate
of the cross traffic is set to 25~Mbps. The probing method
seeks to determine the unused capacity of the  link in
the center of the figure.
The measurements do not address losses of
probe traffic. In fact, when a probe packet is dropped, the
measurement for this packet is ignored.

\begin{figure}
\centering
\subfigure[Rate Scanning.] {\epsfig{file=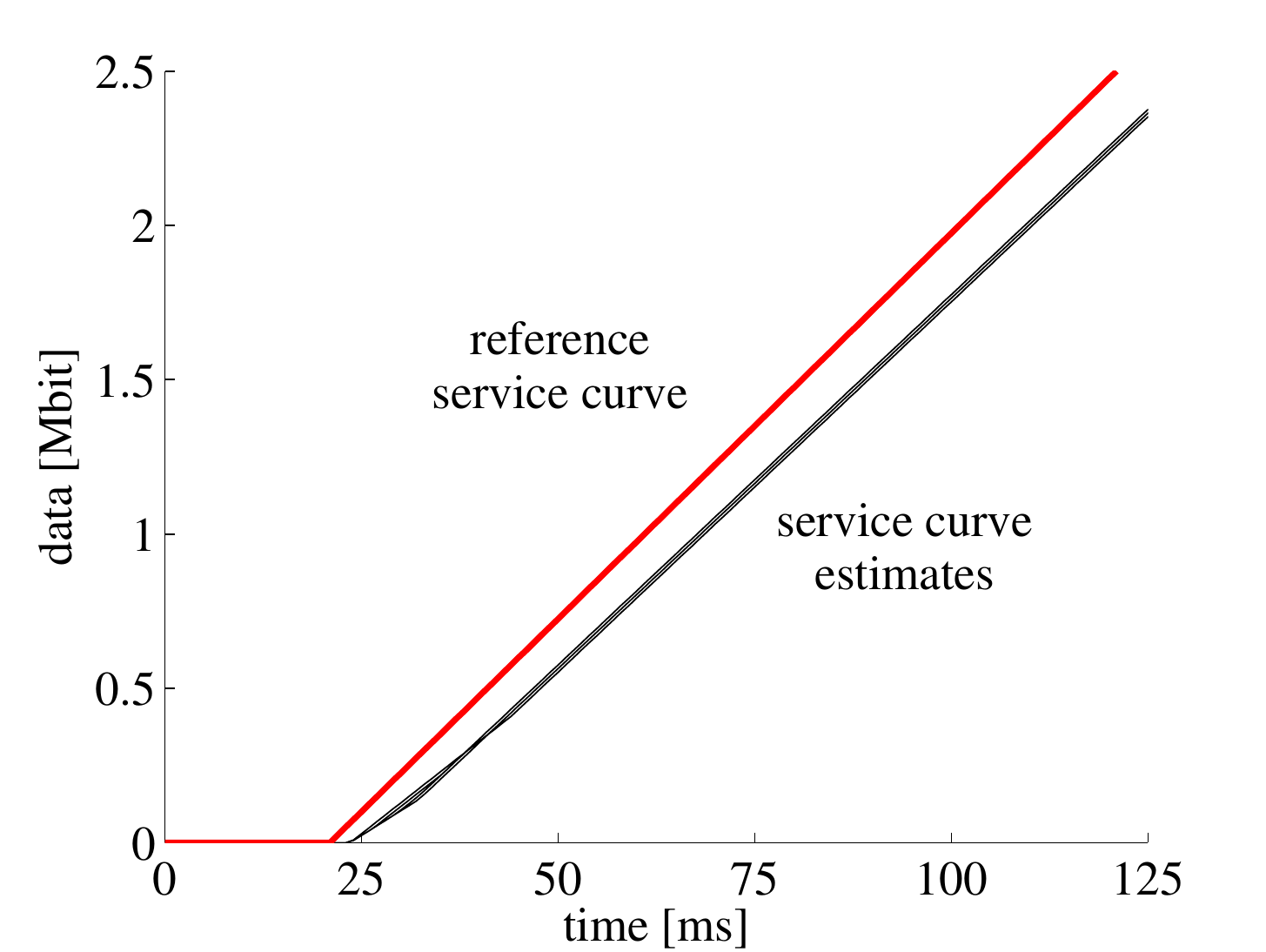,
width=0.49 \linewidth}\label{fig:chirp-vs-scan-a}}
\subfigure[Rate Chirp.] {\epsfig{file=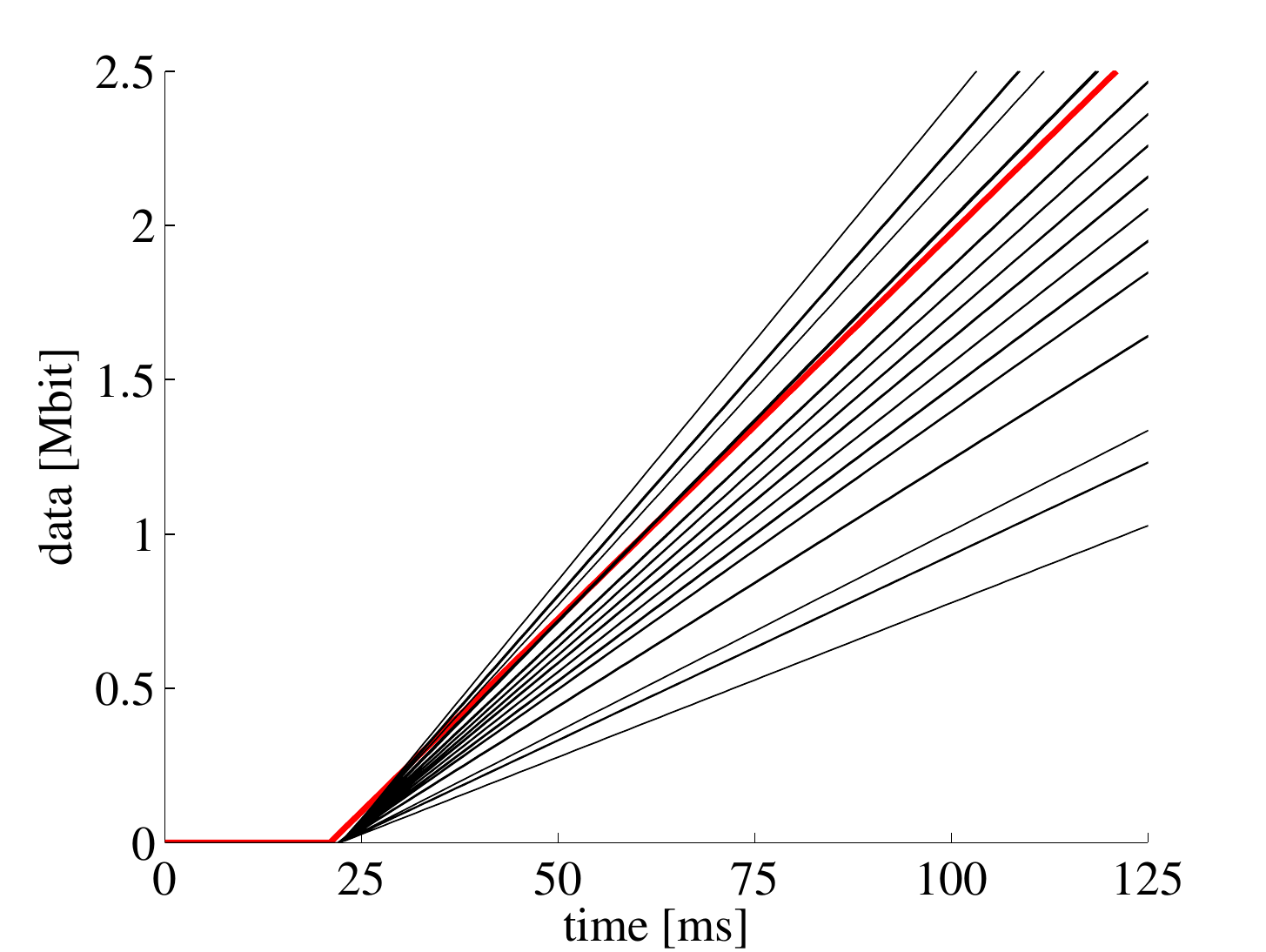,
width=0.49 \linewidth}\label{fig:chirp-vs-scan-b}}
\label{fig:chirp-vs-scan} \caption{Experiment 1: Service curve
estimates with CBR cross traffic.}
\end{figure}

\subsection{Experiment 1: Rate scanning vs. Rate chirps}
We first compare
the effectiveness of the Rate Scanning and Rate Chirp methods from
Sections~\ref{subsec:ratescan} and ~\ref{subsec:ratechirp} in the
dumbbell topology. We assume that CBR cross traffic is sent at a rate of 25~Mbps.

For the rate scanning method, each packet
train has 400 packets, transmitted in increments of 4~Mbps, up to at most 60~Mbps. The stopping criterion is the backlog convexity
criterion from Section~\ref{subsec:stop}
with a threshold of $\alpha = 4$~ms and a window size of $W=3$ for median filtering.

For the rate chirp method, the initial spacing of probe packets is
set to a rate of 4~Mbps, where the spacing between subsequent
packets is governed by spread factor of $\gamma = 1.05$.
The chirp is stopped once its instantaneous rate reaches  100~Mbps,
resulting in 66 packets for each chirp.
The reason we let the rate chirps go up
to 100~Mbps whereas the rate scans only go up to 60 Mbps is that
data points at the end of the chirps become quite sparse due to
the geometric increase of the chirp's rate.
For rate chirps, we employ the stopping criterion proposed in \cite{ribeiro:pathchirp}, which aims at finding the instantaneous data rate at which one way
packet delays start growing due to persistent overload.
(Note that an  application of the non-linearity criterion from Section~\ref{subsec:stop} to the rate chirp method would require multiple rate chirps.)

In Figs.~\ref{fig:chirp-vs-scan-a} and~\ref{fig:chirp-vs-scan-b}
we present the results of 100 repeated estimates of the available bandwidth in
terms of the computed service curves (shown as black graphs) for
the rate scanning and rate chirp method, respectively.
As noted in Section~\ref{sec-probe}, each sample of the available bandwidth can be thought of being conditioned on the state of the network.
We include,
as a red graph, a rate-latency curve with the minimal delay (of
approximately 21~ms)\footnote{The minimal delay consists of 20~ms
propagation delay and approximately 1~ms transmission delay.} and the average available
bandwidth (25~Mbps). This curve is referred to as  {\em reference
service curve} and serves as  an a priori bound for the available bandwidth computations.

A comparison of Figs.~\ref{fig:chirp-vs-scan-a} and~\ref{fig:chirp-vs-scan-b} shows that rate scanning provides more reliable estimates of the service curve than rate chirps.
We note that the  {\em pathchirp} method from \cite{ribeiro:pathchirp}, would yield better results since it   smoothes the available bandwidth over  11 estimates to deal with the variability of estimates from single rate chirps.
Rate scanning and rate chirps  perform equally in an ideal linear time-invariant system, while while rate chirps are more susceptible to random noise.

In the remaining experiments, we only consider the rate scanning
method.

\begin{figure}
\centering
\subfigure[Exponential cross traffic.]
{\epsfig{file=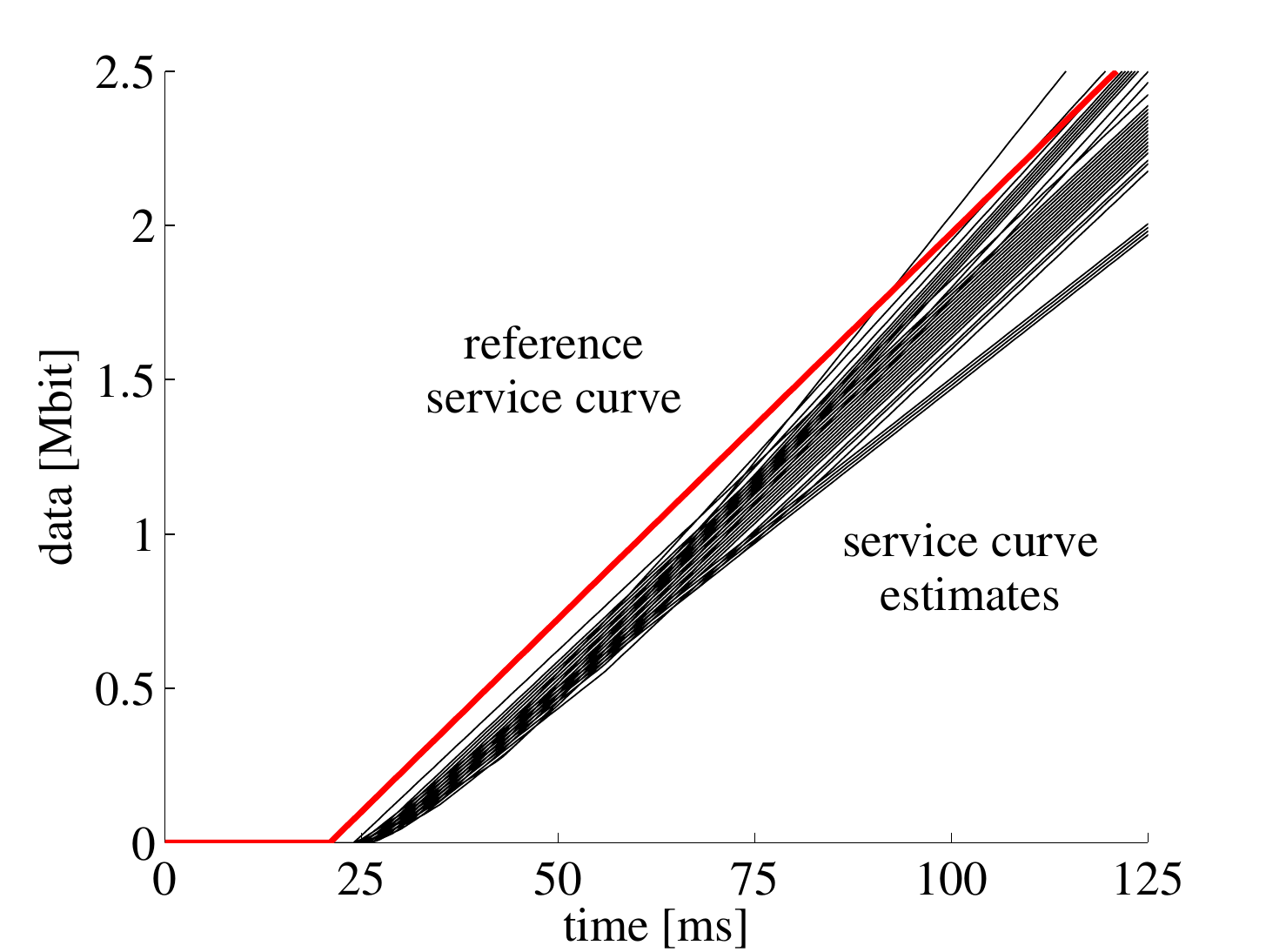, width=0.49
\linewidth}\label{fig:linearity-vs-overload-a}}
\subfigure[Pareto cross traffic.]
{\epsfig{file=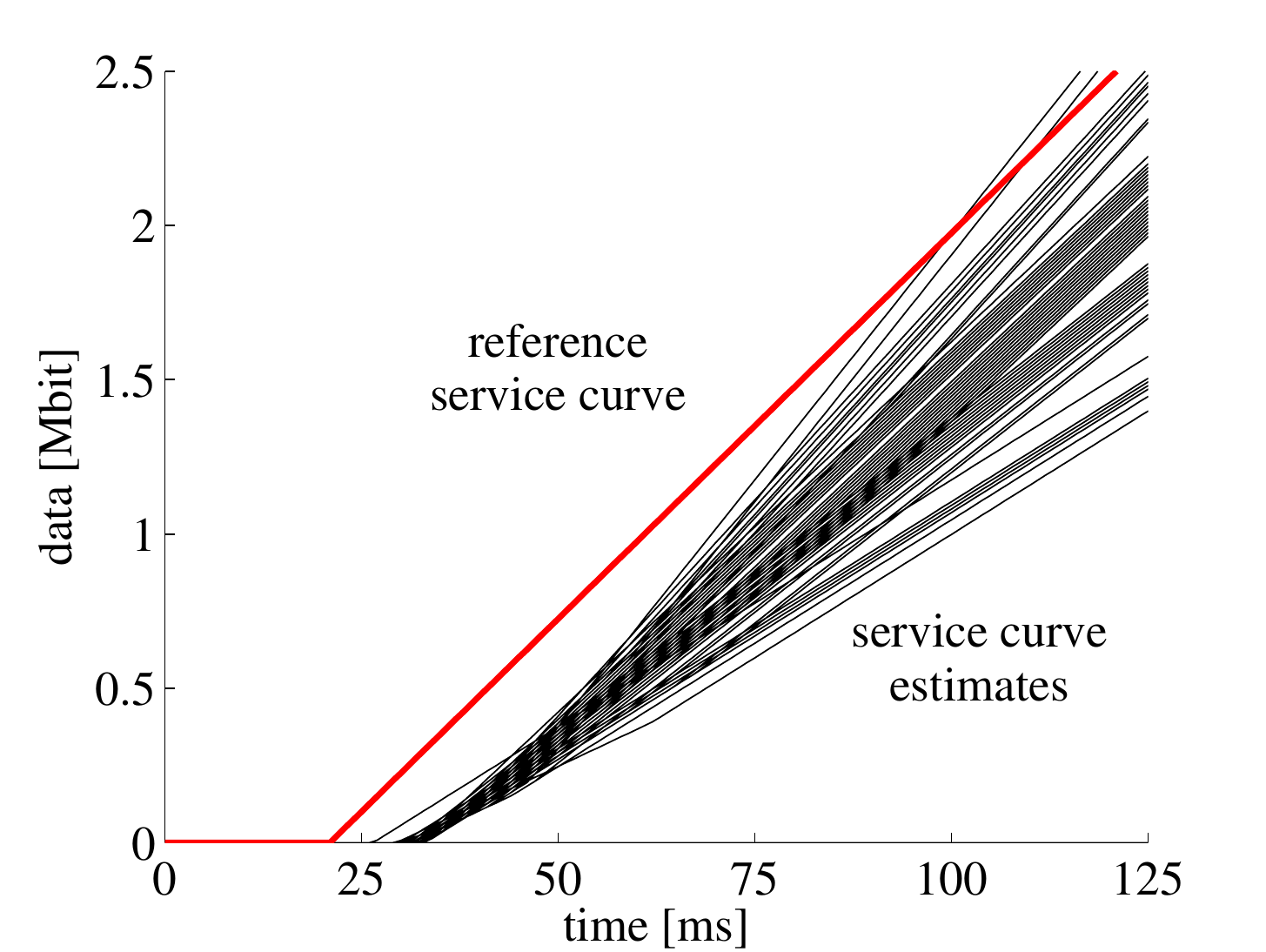, width=0.49\linewidth}  \label{fig:pareto-a}}
\label{fig:linearity-vs-overload} \caption{Experiment 2: Rate scanning with different cross traffic.}
\end{figure}

\subsection{Experiment 2: Different cross traffic distributions.}

In this experiment we evaluate the rate scanning method for
different distributions of the cross traffic on the dumbbell topology.
We consider cross traffic where interarrivals follow an exponential or Pareto distribution (with shape parameter set to 1.5).
All other parameters are as in Experiment~1. In particular, the average traffic rate of cross traffic is 25~Mbps.

In Figs.~\ref{fig:linearity-vs-overload-a} and~\ref{fig:pareto-a}, respectively,  we show the results for  exponential distribution and Pareto cross traffic. The reference service curve is shown in red.
Is is apparent that, compared to CBR cross traffic in Experiment~1, the higher variance
of the cross traffic results in a higher variability of the service
curve estimates. At the same time, even for Pareto traffic, almost all estimates of the available bandwidth provide a conservative bound for the reference  service curve.

\begin{figure}

\centering \epsfig{file=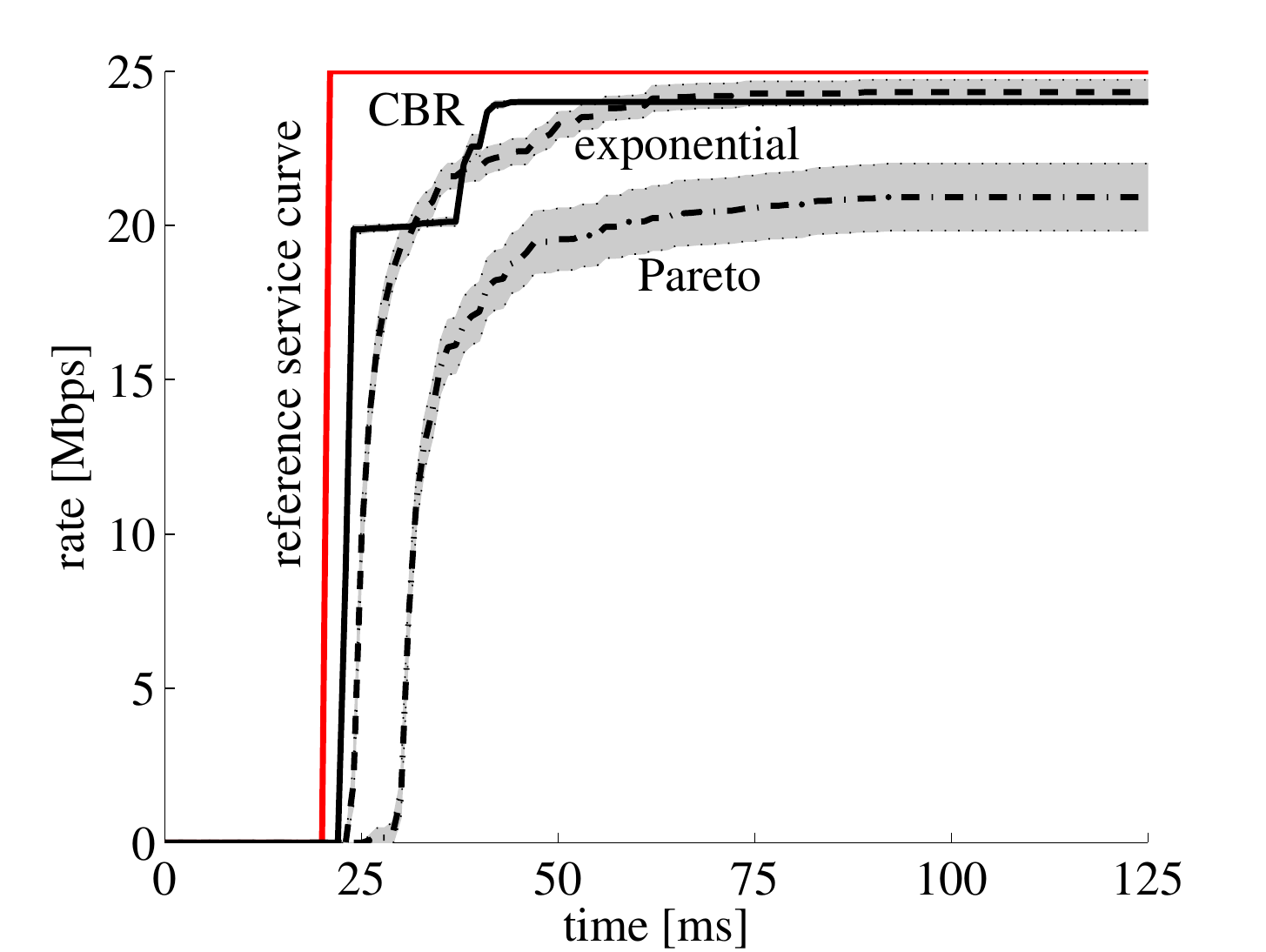,
width=0.7 \linewidth} \caption{Results from Experiments 1--2:
Derivative of service curves.}
 \label{fig:pareto-b}
\end{figure}
\begin{table}[t]
\begin{center}
\begin{tabular}{|l|r|r|}
 \hline
 Cross traffic & lower bound & upper bound \\
 \hline
 CBR  & 22.6 Mbps & 22.8 Mbps \\
 Exponential & 17.7 Mbps & 25.4 Mbps \\
 Pareto & 15.9 Mbps & 29.3 Mbps \\
 \hline
\end{tabular}
\end{center}
\caption{Pathload measurements.}  \label{tab:pathload1}
\end{table}

In Fig.~\ref{fig:pareto-b} we reconcile the results from
Figs.~\ref{fig:chirp-vs-scan-a},~\ref{fig:linearity-vs-overload-a}
and ~\ref{fig:pareto-a} in a single graph. We compute the
derivatives of the service curves and plot the mean value averaged
over the 100~estimates (with 95\%~confidence
intervals). The graph also includes the derivative of the reference service curve (in red). The plot for the reference service
curve shows a sudden increase at time  21~ms, where the service
curve jumps to a rate of 25~Mbps. The derivatives of the
service estimates for CBR, exponential, and Pareto cross traffic provide lower bounds, which become more pessimistic with increasing variance of the cross traffic distribution.

As a point of reference, we now show results of the {\em pathload}
application available from \cite{pathloadsources} for the same
network and cross traffic parameters. {\em Pathload} is frequently
used as a benchmark  to evaluate bandwidth estimation
techniques. The {\em pathload} application views available
bandwidth as a rate and returns a range that is averaged over a time interval $\tau$ which bounds the observed distribution of the available bandwidth. For each cross traffic type, we ran {\em
pathload} 100 times  and computed the average values of the lower and upper bounds of the estimated available bandwidth range. The results are summarized in Table~\ref{tab:pathload1}. A comparison of the table with Fig.~\ref{fig:pareto-b} shows that the lower bounds of the min-plus theoretic estimation yield service curves,  whose long term average rate is similar to or above the lower bound of {\em pathload} measurements. As expected, the variation range increases if the cross-traffic has a higher variability.

\begin{figure}
\centering \epsfig{file=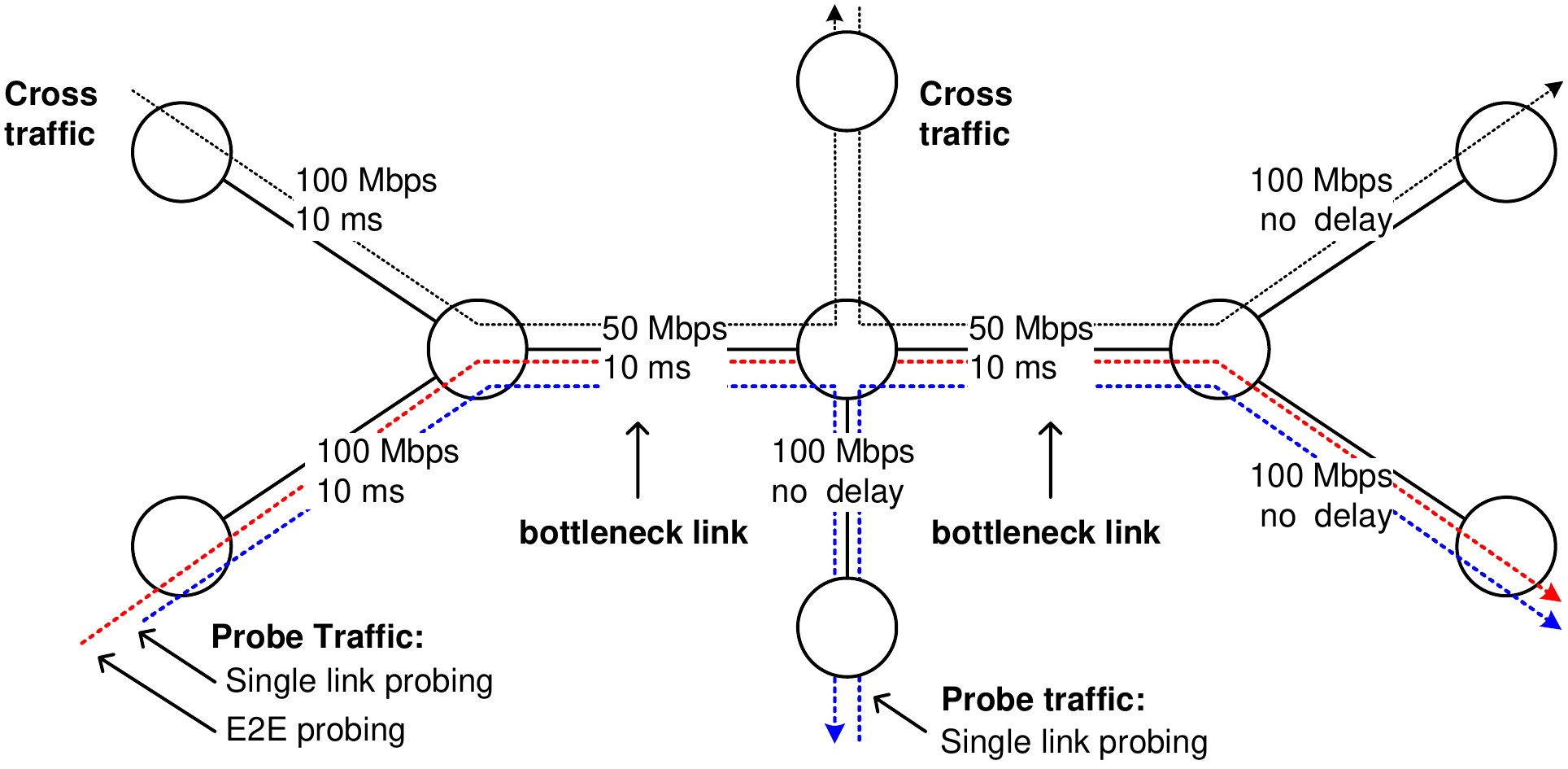, width=0.95 \linewidth}
\caption{Topology with multiple bottleneck links.}
\label{fig:dumbbell-2}
\subfigure[Two  bottleneck links.]
{\epsfig{file=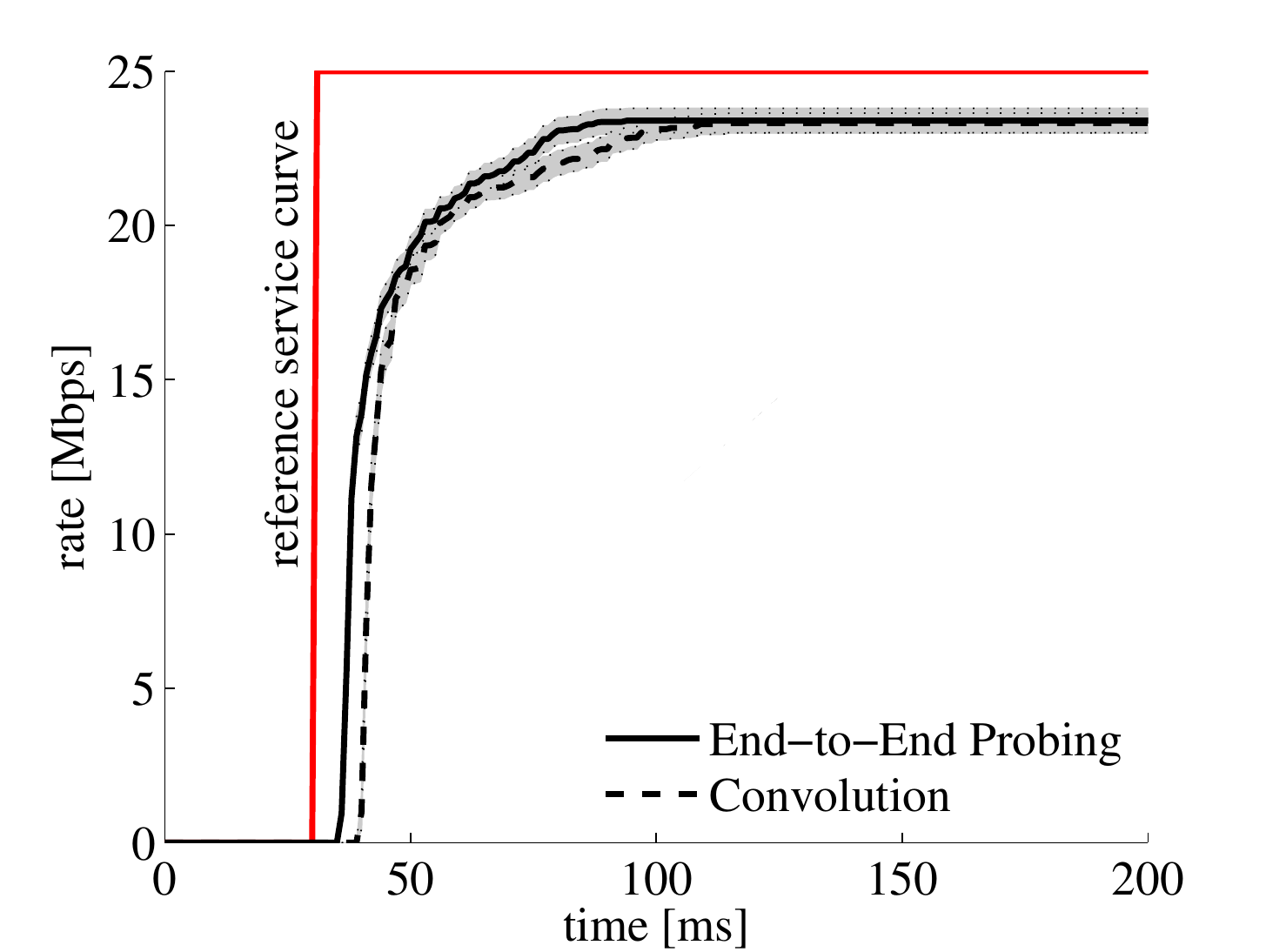, width=0.49
\linewidth}\label{fig:2node-endtoend}}
\subfigure[Three  bottleneck links.]
{\epsfig{file=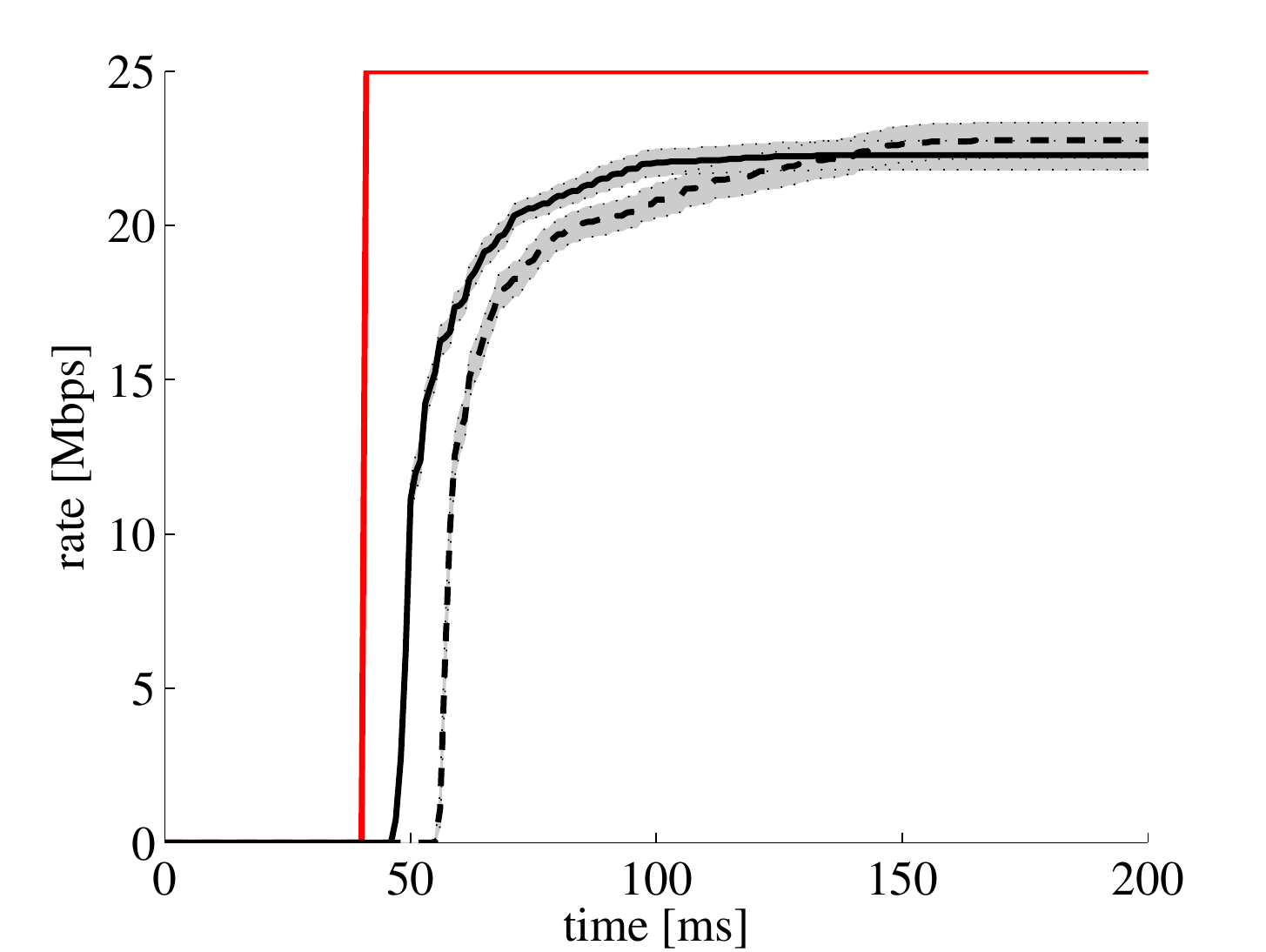, width=0.49
\linewidth}\label{fig:3node-endtoend}}

\centerline{ \subfigure[Four bottleneck links.]
{\epsfig{file=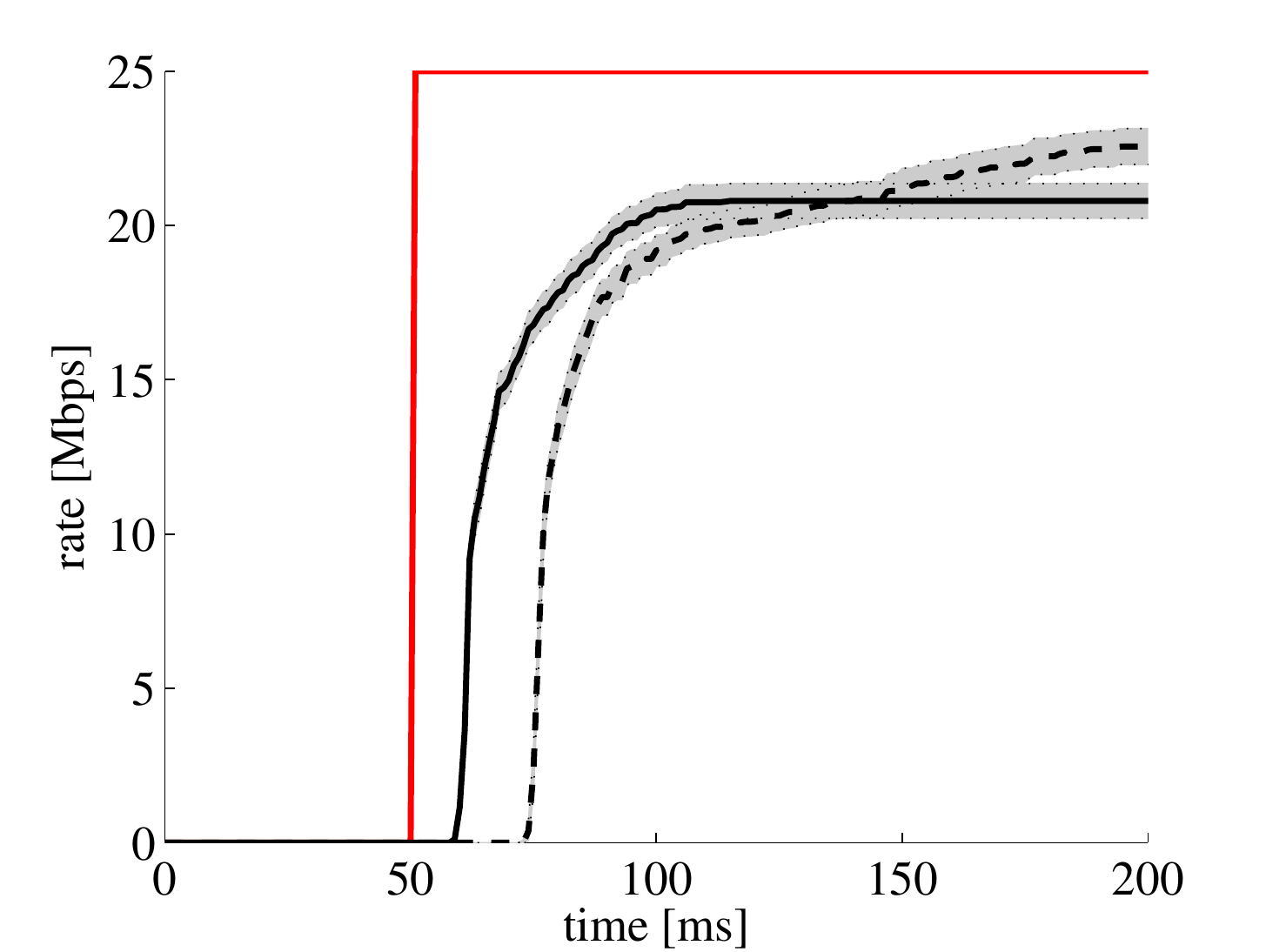, width=0.49
\linewidth}\label{fig:4node-endtoend}} }
\caption{\label{fig:end-to-end}Experiment 4: Derivative of service curves in multi-node
measurements.}

\end{figure}
\begin{table}
\begin{center}
\begin{tabular}{|l|c|c|}
 \hline
No. bottle- & End-to-End & Per-link probing \\
neck links& probing & with Eq.~(\ref{eq:alpha-new}) \\
 \hline
 2 & $[14.3, 20.5]$~Mbps &  $[15.1, 22.6 ]$~Mbps\\
 3 & $[12.2, 18.2]$~Mbps & $[13.3, 19.9]$~Mbps\\
 4 & $[11.7, 17.1]$~Mbps & $[11.8, 18.1]$~Mbps\\
 \hline
\end{tabular}
\end{center}
\caption{Pathload measurements: Multiple Bottlenecks.}
\label{tab:pathload2}
\end{table}

\subsection{Experiment 3: Multiple bottleneck links}

We now present measurements over networks with multiple  bottleneck links. Fig.~\ref{fig:dumbbell-2} depicts the network setup in Emulab with two bottleneck links. The bottleneck links have a capacity of 50~Mbps. The interarrival distribution of cross traffic is  exponential with
parameters as discussed earlier in this section.
As probing scheme, we again use rate scanning with the backlog convexity stopping criterion.

For each network, we compute the end-to-end service curve using two
 methods. In the first method, called {\em End-to-End (E2E)
Probing}, we send probe traffic end-to-end over all bottleneck links. In
the second method, referred to as {\em Convolution} we send probe traffic separately over bottleneck link and construct a service curve for each link. Then, we compute
an end-to-end service curve using the convolution operation following Eq.~(\ref{eq-conv}).
The convolution can be done efficiently in the Legendre domain, where the convolution becomes  a simple addition.

Note that in our computation of the available bandwidth over multiple links, the convolution from Eq.~(\ref{eq-conv}) replaces the minimum in the widely used Eq.~(\ref{eq:alpha-new}).
For the special case that the available bandwidths of links are constant-rate functions
the convolution over multiple links is equal to the minimum of the rates.  Formally, if $S_i (t) = r_i t$ for all $i$, we obtain
$S(t) = S_1 \conv S_2 \conv \ldots \conv S_N (t) = \min_i r_i$.
Thus, the convolution expression is a true generalization of the currently prevailing method for composing bandwidth estimation
of multiple links.

In Figs.~\ref{fig:2node-endtoend}--\ref{fig:4node-endtoend},
we present the outcomes of our experiments for two, three, and four bottleneck links, respectively. As in Fig.~\ref{fig:pareto-b}, we
present  derivatives of the service curves. We depict the
average values of 100 measurements, as well as the 95\%~confidence
intervals. The
 reference service curve (in red) is a  latency rate
service curve with a delay of 10~ms for each traversed bottleneck link and a rate equal to the average unused link capacity.
We observe that the results of E2E probing are lower at shorter time scales. Over longer time intervals, the results of
the Convolution method yields larger estimates.
The  long-term average rate of the computed service curves degrades with the number of hops.

In Table~\ref{tab:pathload2} we show as benchmark the results of {\em pathload} measurements.  We include the range of values of end-to-end probing, as well as the results of applying Eq.~(\ref{eq:alpha-new}) to per-link measurements. The degradation of available  bandwidth estimates as the number of bottleneck links is increased is similar as observed in Fig.~\ref{fig:end-to-end}. The long-term average of the service curve  in Fig.~\ref{fig:end-to-end} yields more
optimistic results than the range of values in Table~\ref{tab:pathload2}. While the results of the system theoretic approach for paths with multiple nodes are clearly encouraging,  we caution against a
generalization to other topologies and production networks.

\nocite{fidler07}

\clearpage
\section{Conclusions}
\label{sec-concl}

We have developed an interpretation of bandwidth estimation as
a problem in min-plus linear systems, where the available bandwidth
is represented by a service curve. Using service curves as opposed to constant-rate functions permits a description of  bandwidth availability at different time-scales.
We have related difficulties with network probing to non-linearities of the underlying system.
By interpreting a network as a system that is min-plus linear at low loads, and becomes non-linear when the network load exceeds a
threshold, we have argued that the crossing of the linear and non-linear regions marks the point  where the available bandwidth  can be observed.
A series of measurement experiments showed that the min-plus linear approach to bandwidth estimation lends itself to the development
of effective probing schemes. In particular, the min-plus convolution operator can be applied to obtain end-to-end estimates from per-link measurements.

\section*{Acknowledgments}
The authors thank A. Burchard for many insights and
suggestions.

\small{
\renewcommand{\baselinestretch}{1.1}
}


\begin{thebibliography}{10}

\bibitem{ns2}
ns-2 network simulator.
\newblock http://www.isi.edu/nsnam/ns/.

\bibitem{agharebparast:slopedomain}
F.~Agharebparast and V.~C.~M. Leung.
\newblock Slope domain modeling and analysis of data communication networks: A
  network calculus complement.
\newblock In {\em Proc. {IEEE} {ICC}}, pages 591--596, June 2006.

\bibitem{agrawal:flowcontrolprotocols}
R.~Agrawal, R.~L. Cruz, C.~Okino, and R.~Rajan.
\newblock Performance bounds for flow control protocols.
\newblock 7(3):310--323, June 1999.

\bibitem{akella:bfind}
A.~Akella, S.~Seshan, and A.~Shaikh.
\newblock An empirical evaluation of wide-area internet bottlenecks.
\newblock In {\em Proc. {ACM} {IMC}}, pages 101--114, 2003.

\bibitem{baccelli:synchronizationlinearity}
F.~Baccelli, G.~Cohen, G.~J. Olsder, and J.-P. Quadrat.
\newblock {\em Synchronization and Linearity: An Algebra for Discrete Event
  Systems}.
\newblock John Wiley \& Sons Ltd., West Sussex, Great Britain, 1992.

\bibitem{baccelli06}
F.~Baccelli, S.~Machiraoju, D.~Veitch, and J.~Bolot.
\newblock The role of {PASTA} in network measurement.
\newblock In {\em Proc. {ACM} {SIGCOMM}}, pages 231--242, Pisa, September 2006.

\bibitem{cprobe}
R.~L. Carter and M.~E. Crovella.
\newblock Measuring bottleneck link speed in packet switched networks.
\newblock {\em Performance Evaluation}, 27 and 28:297--318, 1996.

\bibitem{cetinkaya:egressadmissioncontrol}
C.~Cetinkaya, V.~Kanodia, and E.~W. Knightly.
\newblock Scalable services via egress admission control.
\newblock 3(1):69--81, March 2001.

\bibitem{chang:performanceguarantees}
C.-S. Chang.
\newblock {\em Performance Guarantees in Communication Networks}.
\newblock Springer-Verlag, London, Great Britain, 2000.

\bibitem{Cruz91}
R.~Cruz.
\newblock A calculus for network delay, parts {I} and {II}.
\newblock {\em IEEE Transactions on Information Theory}, 37(1):114--141,
  January 1991.

\bibitem{cruz:qualityofserviceguarantees}
R.~L. Cruz.
\newblock Quality of service guarantees in virtual circuit switched networks.
\newblock 13(6):1048--1056, August 1995.

\bibitem{Reissleintrace}
G.~Van der Auwera, P.~T. David, and M.~Reisslein.
\newblock Bit rate-variability of h.264/avc frext.
\newblock Technical report, Arizona State University, April 2006.

\bibitem{pathloadsources}
C.~Dovrolis and M.~Jain.
\newblock Pathload: A measurement tool for the available bandwidth of network
  paths.
\newblock http://www.cc.gatech.edu/fac/Constantinos.Dovrolis/pathload.html.

\bibitem{dovrolis:packetdispersion}
C.~Dovrolis, P.~Ramanathan, and D.~Moore.
\newblock What do packet dispersion techniques measure?
\newblock In {\em Proc. {IEEE} {INFOCOM}}, pages 905--914, April 2001.

\bibitem{fidler:legendre}
M.~Fidler and S.~Recker.
\newblock Conjugate network calculus: A dual approach applying the legendre
  transform.
\newblock {\em Computer Networks}, 50(8):1026--1039, June 2006.

\bibitem{hisakado:legendre}
T.~Hisakado, K.~Okumura, V.~Vukadinovic, and L.~Trajkovic.
\newblock Characterization of a simple communication network using legendre
  transform.
\newblock In {\em Proc. Internationl Symposium on Circuits and Systems
  {(ISCAS)}}, pages 738--741, May 2003.

\bibitem{hu:IGI}
N.~Hu and P.~Steenkiste.
\newblock Evaluation and characterization of available bandwidth probing
  techniques.
\newblock 21(6):879--894, August 2003.

\bibitem{jacobson88}
V.~Jacobson.
\newblock Congestion avoidance and control.
\newblock In {\em Proc. {ACM} {SIGCOMM}}, pages 314--329, August 1988.

\bibitem{jain:slops}
M.~Jain and C.~Dovrolis.
\newblock End-to-end available bandwidth: Measurement methodology, dynamics,
  and relation with {TCP} throughput.
\newblock In {\em Proc. {ACM} {SIGCOMM}}, pages 295--308, October 2002.

\bibitem{jain:pathload}
M.~Jain and C.~Dovrolis.
\newblock Pathload: A measurement tool for end-to-end available bandwidth.
\newblock In {\em Proc. Passive and Active Measurement (PAM) Workshop}, pages
  14--25, March 2002.

\bibitem{jain:bandwidthestimationpitfalls}
M.~Jain and C.~Dovrolis.
\newblock Ten fallacies and pitfalls on end-to-end available bandwidth
  estimation.
\newblock In {\em Proc. {ACM} {IMC}}, pages 272--277, 2004.

\bibitem{jain:pathvar}
M.~Jain and C.~Dovrolis.
\newblock End-to-end estimation of the available bandwidth variation range.
\newblock In {\em Proc. {ACM} {SIGMETRICS}}, pages 265--276, June 2005.

\bibitem{Jiang05}
Y.~Jiang, P.J. Emstad, A.~Nevin, V.~Nicola, and M.~Fidler.
\newblock Measurement-based admission control for a flow-aware network.
\newblock In {\em Proc. of 1st EuroNGI Conference on Next Generation Internet
  Networks Traffic Engineering}, pages 318-- 325, April 2005.

\bibitem{loguinov04}
S.-R. Kang, X.~Liu, M.~Dai, and D.~Loguinov.
\newblock Packet-pair bandwidth estimation: Stochastic analysis of a single
  congested node.
\newblock In {\em Proc. IEEE ICNP}, pages 316--325, October 2004.

\bibitem{kapoor:capprobe}
R.~Kapoor, L.-J. Chen, L.~Lao, M.~Gerla, and M.~Y. Sanadidi.
\newblock {CapProbe} a simple and accurate capacity estimation technique.
\newblock In {\em Proc. {ACM} {SIGCOMM}}, pages 67--78, August/September 2004.

\bibitem{keshav91}
S.~Keshav.
\newblock A control-theoretic approach to flow control.
\newblock In {\em Proc. {ACM} {SIGCOMM}}, pages 3--15, September 1991.

\bibitem{rude-crude}
J.~Laine, S.~Saaristo, and R.~Prior.
\newblock Rude/crude.
\newblock http://rude.sourceforge.net/.

\bibitem{padhye04}
K.~Lakshminarayanan, V.~N. Padmanabhan, and J.~Padhye.
\newblock Bandwidth estimation in broadband access networks.
\newblock In {\em Proc. {ACM} {IMC}}, pages 314--321, October 2004.

\bibitem{leboudec:networkcalculus}
J.-Y. {Le Boudec} and P.~Thiran.
\newblock {\em Network Calculus A Theory of Deterministic Queuing Systems for
  the Internet}.
\newblock Springer-Verlag, Berlin, Germany, 2001.

\bibitem{fidler07}
J.~Liebeherr, M.~Fidler, and S.~Valaee.
\newblock Min-plus system interpretation of bandwidth estimation.
\newblock In {\em Proc. {IEEE} {INFOCOM}}, May 2007.

\bibitem{liu:packetpairdispersion}
X.~Liu, K.~Ravindran, and D.~Loguinov.
\newblock What signals do packet-pair dispersions carry?
\newblock In {\em Proc. {IEEE} {INFOCOM}}, pages 281--292, March 2005.

\bibitem{loguinov07}
X.~Liu, K.~Ravindran, and D.~Loguinov.
\newblock A queuing-theoretic foundation of available bandwidth estimation:
  Single-hop analysis.
\newblock {\em IEEE/ACM Transactions on Networking}, 15(4):918--931, August
  2007.

\bibitem{loguinov08}
X.~Liu, K.~Ravindran, and D.~Loguinov.
\newblock A stochastic foundation of available bandwidth estimation: Multi-hop
  analysis.
\newblock {\em IEEE/ACM Transactions on Networking}, 16(2), April 2008.
\newblock (To appear).

\bibitem{baccelli07}
S.~Machiraju, D.~Veitch, F.~Baccelli, and J.~Bolot.
\newblock Adding definition to active probing.
\newblock {\em ACM SIGCOMM Computer Communication Review}, 37(2):19--28, April
  2007.

\bibitem{topp}
B.~Melander, M.~Bj{\"o}rkman, and P.~Gunningberg.
\newblock A new end-to-end probing and analysis method for estimating bandwidth
  bottlenecks.
\newblock In {\em Proc. {IEEE} {GLOBECOM}}, pages 415--420, November 2000.

\bibitem{melander:fcfsprobing}
B.~Melander, M.~Bj{\"o}rkmann, and P~Gunningberg.
\newblock First-come-first-served packet dispersion and implications for {TCP}.
\newblock In {\em Proc. {IEEE} {GLOBECOM}}, pages 2170--2174, November 2002.

\bibitem{naudts:trafficparametermeasurement}
J.~Naudts.
\newblock Towards real-time measurement of traffic control parameters.
\newblock {\em Computer Networks}, 34(1):157--167, July 2000.

\bibitem{ABwE}
J.~Navratil and R.~L. Cottrell.
\newblock {ABwE}: A practical approach to available bandwidth estimation.
\newblock In {\em Proc. Passive and Active Measurement (PAM) Workshop}, pages
  1--11, April 2003.

\bibitem{paxson:measurements}
V.~Paxson.
\newblock {\em Measurements and Analysis of End-to-End Internet Dynamics}.
\newblock PhD thesis, Univ. of California, Berkeley, April 1997.

\bibitem{ribeiro:pathchirp}
V.~Ribeiro, R.~Riedi, R.~Baraniuk, J.~Navratil, and L.~Cottrell.
\newblock {PathChirp}: Efficient available bandwidth estimation for network
  paths.
\newblock In {\em Proc. Passive and Active Measurement Workshop}, April 2003.

\bibitem{rockafellar:convexanalysis}
R.~T. Rockafellar.
\newblock {\em Convex Analysis}.
\newblock Princeton University Press, 1972.

\bibitem{shriram07}
A.~Shriram and J.~Kaur.
\newblock Empirical evaluation of techniques for measuring available bandwidth.
\newblock In {\em Proc. {IEEE} {INFOCOM}}, pages 2162--2170, May 2007.

\bibitem{shriam05}
A.~Shriram, M.~Murray, Y.~Hyun, N.~Brownlee, A.~Broido, M.~Fomenkov, and K.~C.
  Claffy.
\newblock Comparison of public end-to-end bandwidth estimation tools on
  high-speed links.
\newblock In {\em Proc. Passive and Active Measurement Workshop}, pages
  306--320, March 2005.

\bibitem{sommers}
J.~Sommers, P.~Barford, and W.~Willinger.
\newblock Laboratory-based calibration of available bandwidth estimation tools.
\newblock {\em Microprocessors and Microsystems}, 31(4):222--235, June 2007.

\bibitem{strauss:spruce}
J.~Strauss, D.~Katabi, and F.~Kaashoek.
\newblock A measurement study of available bandwidth estimation tools.
\newblock In {\em Proc. {ACM} {IMC}}, pages 39--44, 2003.

\bibitem{BinTariq}
M.~M.~Bin Tariq, A.~Dhamdhere, C.~Dovrolis, and M.~Ammar.
\newblock Poisson versus periodic path probing (or, does {PASTA} matter?).
\newblock In {\em Proc. 5th Conference on Internet Measurement}, pages
  119--124, October 2005.

\bibitem{valaee:adhocadmissioncontrol}
S.~Valaee and B.~Li.
\newblock Distributed call admission control for ad hoc networks.
\newblock In {\em Proc. of IEEE 56th Vehicular Technology Conference (VTC 2002-Fall)}, pages 1244--1248, September 2002.

\bibitem{emulab}
B.~White and et. al.
\newblock An integrated experimental environment for distributed systems and
  networks.
\newblock In {\em Proc. of OSDI 2002}, pages 255--270, December 2002.

\bibitem{Wu03}
D.~Wu and R.~Negi.
\newblock Effective capacity: A wireless link model for support of quality of
  service.
\newblock {\em IEEE Transactions on Wireless Communications}, 2(4):630--643,
  July 2003.

\bibitem{yzhang}
Y.~Zhang, N.~Duffield, V.~Paxson, and S.~Shenker.
\newblock On the constancy of internet path properties.
\newblock In {\em Proc. ACM SIGCOMM Internet Measurement Workshop}, pages
  197--211, November 2001.

\end{thebibliography}
\end{document}